\documentclass[journal,twoside,web]{ieeecolor}

\usepackage{generic}
\usepackage{cite}
\usepackage{amsmath,amssymb,amsfonts}
\usepackage{algorithmic}
\usepackage{graphicx}
\usepackage{algorithm,algorithmic}
\usepackage{hyperref}
\let\labelindent\relax
\usepackage{enumitem}
\usepackage{tikz}
\usepackage{textcomp}
\hypersetup{hidelinks=true}

\def\BibTeX{{\rm B\kern-.05em{\sc i\kern-.025em b}\kern-.08em
    T\kern-.1667em\lower.7ex\hbox{E}\kern-.125emX}}
\usetikzlibrary{positioning}
\usetikzlibrary{shapes.geometric, arrows}
\usetikzlibrary{arrows.meta}

\tikzstyle{startstop} = [rectangle, rounded corners, minimum width=3cm, minimum height=1cm, text centered, draw=black, fill=red!20]

\tikzstyle{process} = [rectangle, minimum width=3cm, minimum height=1cm, text centered, draw=black, fill=blue!20]

\tikzstyle{decision} = [diamond, aspect=2, text width=3cm, text centered, draw=black, fill=green!20]

\tikzstyle{arrow} = [thick, ->, >=stealth]

\newtheorem{proposition}{Proposition}
\newtheorem{theorem}{Theorem}
\newtheorem{corollary}{Corollary}
\newtheorem{lemma}{Lemma}

\begin{document}
%

\title{Minimizing the Expected Cost of Synchronization in Lossless Power Networks}

\author{Gerald~Ogbonna,~\IEEEmembership{Student Member,}
        David~Bindel,~\IEEEmembership{Senior Member,}
        C. L.~Anderson,~\IEEEmembership{Senior Member}
\thanks{G. Ogbonna is with the Department of Systems Engineering, Cornell University, Ithaca, NY, 14853 USA (e-mail: gco27@cornell.edu).}
\thanks{C. L. Anderson is with the Department of Biological and Environmental Engineering, Systems Engineering, and the Center for Applied Mathematics, Cornell University, Ithaca, NY, 14853 USA (email: cla28@cornell.edu).}
\thanks{David Bindel is with the Department of Computer Science, Mathematics, and the Center for Applied Mathematics, Cornell University, Ithaca, NY, 14853 USA (email: bindel@cornell.edu).}
\thanks{Authors would like to thank the Cornell Atkinson Center for Sustainability and the National Science Foundation for support of this work.\\ Manuscript received XX, 2025; revised XX.}
}

\maketitle

\begin{abstract}
The reliable operation of large-scale electric power networks is increasingly challenging, particularly with the integration of stochastic renewable generation. In this work, we address the problem of minimizing network transients (cost of synchronization) by optimally modifying the underlying network. We formulate the problem in terms of graph Laplacian matrices and show that, under typical assumptions, the problem is convex. We derive a linear matrix inequality whose feasibility guarantees the existence and uniqueness of phase cohesive steady-state angles; this condition can be directly incorporated as a convex constraint in the optimization framework and we provide geometric interpretations of the problem. The proposed method is validated on the IEEE 30-bus test system, where results demonstrate that our approach effectively identifies critical links on the network. Simulations of the network dynamics show a substantial reduction in network transients and overall improvements across several performance metrics.

\end{abstract}


%
\IEEEpeerreviewmaketitle

\section{Introduction}
\IEEEPARstart{T}{he} electric power system is currently undergoing rapid transformation with the increasing integration of renewable generation. It is well known that reliable power system operation depends strongly on the synchronization of multiple oscillatory state variables to a nominal grid frequency \cite{paganini_global_2020, chiang_direct_2011}. The stochasticity of renewables poses new challenges to the ability of the power network to remain synchronized and return to synchrony following a perturbation from steady-state. \cite{ersdal_model_2014} showed that the number of frequency incidents per month on the Nordic system is strongly correlated with the share of generation from renewables, highlighting the need for fast acting controllers and the design of resilient/robust networks. Designing resilient networks provides several benefits to system operations including ensuring that the system returns to synchrony following a disturbance with minimal control effort (i.e. reduce control burden on local controllers), improved reference tracking, reduce equipment damage due to large deviations in the system frequency, and prevent cascaded failures.

Previous studies \cite{paganini_global_2017, guo_graph_2018, mallada_improving_2011}, have considered the role of network connectivity on the transient stability of power networks. In \cite{paganini_global_2017, klein_fundamental_1991}, the authors observe that geographically distributed frequency oscillations (inter-area oscillations) are often the result of weak global coupling. Similarly, Low \textit{et al.} \cite{guo_critical_2017} suggest that increasing network connectivity increases system robustness against disturbances of relatively low frequency. Overall, there is consensus in the literature that improving network connectivity enhances system robustness to disturbances. This observation motivates the present work.

The problem of designing robust networks has been a focus of recent research efforts \cite{wang_improving_2014, nagpal_designing_2022, abbas_robust_2012}. To this end, several measures of network robustness \cite {freitas_graph_2022, siami_schur-convex_2014} have been defined in the literature for different classes of dynamical systems on graphs. For networks in which all nodes have identical second-order Kuramoto oscillators, Tyloo \textit{et al.} \cite{tyloo_global_2019} provide robustness measures defined in terms of the generalized effective resistance of the underlying graph. Specifically, they show a strong correlation between their measures and the $\mathcal{L}_{2}$ norm of the phase and frequency trajectories of the network dynamics for perturbations following an Ornstein-Uhlenbeck process. This measure is minimized in the work by Nagpal \textit{et al.} \cite{nagpal_designing_2022}.

The use of effective resistance as a measure of network connectivity (graph connectivity/network robustness) dates to early works by Klein and Randic \cite{klein_resistance_nodate}. Ghosh \textit{et al.} \cite{ghosh_minimizing_2008} provide a more recent review of various applications of effective resistance to graph theoretic problems. Applications of effective resistance in network analysis include growing linear consensus networks \cite{siami_tractable_2016}, community detection\cite{pizzuti_effective_2021}, power redispatch and network synchronization \cite{fazlyab_optimal_2017}, improving robustness under random time-varying topological network changes \cite{dai_optimal_2011}, and designing fast mixing Markov processes\cite{sun_fastest_2006}.

In this paper, we investigate the problem of designing robust power networks with the aim of minimizing transients by optimizing the connectivity of the transmission system. We model the power grid using a coupled-oscillator framework connected over an undirected graph and formulate optimization problems to improve the network's dynamic response to disturbances. The changes to the transmission network proposed in this work may be implemented through flexible AC transmission systems (FACTS) devices \cite{hingorani_flexible_1993, gotham_power_1998}, which allow the electrical characteristics of transmission lines to be modified, or through strategic expansions of certain transmission corridors. 


The main contributions of this work are as follows:
\begin{itemize}
    \item We extend the results in \cite{paganini_global_2020} to show that the expected $\mathcal{L}_2$ norm of the transient component of the angular frequencies near a frequency-synchronized solution is a function of the \textit{total effective resistance} of the \textit{Kron-reduced graph}.
    \item We show that under certain assumptions, optimizing the network susceptances to minimize the expected deviation from synchrony for a network of homogeneous second-order oscillators is a convex problem and admits a semidefinite programming (SDP) formulation, enabling efficient numerical solutions.
    \item We provide a sufficient condition for steady-state phase angle cohesion and a linear matrix inequality (LMI) whose feasibility ensures the existence and uniqueness of synchronized solutions that meet prescribed steady-state phase angle requirements for a set of net power injections.
    \item We test our methods on the IEEE 30-bus network and show that the method identifies the critical links for synchronization. Simulations of the network dynamics also show that the optimized networks have significantly improved transient response.
\end{itemize}

The rest of the paper is organized as follows: Section II introduces notations, Kron reduction and effective resistance, and the power system model. In Section III we derive a measure of the dynamic performance of the network in terms of effective resistance. In Section IV, we formulate several convex optimization problems that minimize the defined measure and incorporate synchronization constraints. Section V presents numerical results on the test system, and Section VI concludes the paper.


\section{Preliminaries}
\subsection{Notation}
\label{subsec:notation}
Given a matrix $A \in \mathbb{R}^{n \times n}$, the notation $A \succeq 0$ ($A \succ 0$) denotes that the matrix $A$ is symmetric positive semidefinite (positive definite). Similarly, given matrices $A, B \in \mathbb{R}^{n \times n}$, the notation $A \succeq B$ implies $A - B \succeq 0$. For vectors $x, y \in \mathbb{R}^n$, the notation $x \succeq y$ denotes an elementwise partial ordering,
that is $x_i \geq y_{i}$ for each $i$. $e_i$ denotes the $i$th canonical basis vector in $\mathbb{R}^n$ and we use $\mathbf{1}$ and $\mathbf{0}$ to denote the $n$-dimensional vectors of all ones and all zeros, respectively. $v_i$ or $[v]_i$ to denotes the $i$th entry of the vector $v$.

We denote by $\mathcal{G} = (\mathcal{V}, \mathcal{E}, A)$ an undirected graph with vertex set $\mathcal{V} = \{1, \ldots, n\}$ of cardinality $|\mathcal{V}| = n$, $\mathcal{E} \subseteq \mathcal{V} \times \mathcal{V}$ denotes the edge set of cardinality $|\mathcal{E}| = m$, and $A = A^{\intercal} \in \mathbb{R}^{n \times n}$ is the corresponding weighted adjacency matrix. We denote the edge weight of the undirected edge $\{i,j\} \in \mathcal{E}$, by $a_{ij} \geq 0$. $L$ is the Laplacian matrix defined as $L = D - A$, $D$ is the diagonal degree matrix with $D_{ii} = \sum_{j}a_{ij}$. $\text{Spec}(L)$ denotes the multi-set of eigenvalues of $L$. In this work, we assume that the graph $\mathcal{G}$ is connected which implies that the algebraic multiplicity of the eigenvalue $0 \in \text{Spec}(L)$ is $1$, $\text{Ker}(L) = \text{Span}(\mathbf{1})$, and $\text{Range}(L) = \mathbf{1}^{\perp}$, the orthogonal complement of $\text{Span}(\mathbf{1})$. 

We denote an arbitrarily oriented incidence matrix by $B \in \mathbb{R}^{n \times m}$ defined element-wise as $B_{kl} = 1$ if node $k$ is a sink for the oriented edge $l$, $-1$ if node $k$ is a source for the edge $l$, and $0$ otherwise. If $W \in \mathbb{R}^{m \times m}$ denotes the diagonal matrix of edge weights, the Laplacian can be defined in terms of $B$ and $W$ as $L = B WB^{\intercal}$. The Moore-Penrose pseudoinverse of $L$ is denoted by $L^{\dagger}$ and a Regularized Laplacian with regularization parameter $\beta$ by $L_{\text{reg}, \beta}$. $\Pi_n = I - \frac{1}{n}\mathbf{11}^{\intercal}$ denotes the orthogonal projector on to $\mathbf{1}^{\perp}$. We use $x(t)$ to denote time domain variables and $x(s)$ to denote Laplace/frequency domain variables.

\subsection{Kron Reduction and Effective Resistance}
\label{subsec:kron_reduction_and_effective_resistance}

Given a graph $\mathcal{G}$, whose Laplacian matrix $L \in \mathbb{R}^{n \times n}$, and a set of indices $\mathcal{V}_G \subset \mathcal{V} = \{1, \ldots \, n\}$, where $\overline{\mathcal{V}}_G$ denotes the complement of the set $\mathcal{V}_G$, the matrix $L$ can be partitioned as
\begin{align*}
    L = \begin{bmatrix}
        L_{\mathcal{V}_G \mathcal{V}_G} & L_{\mathcal{V}_G \overline{\mathcal{V}}_G}\\ L_{\overline{\mathcal{V}}_G\mathcal{V}_G} & L_{\overline{\mathcal{V}}_G\overline{\mathcal{V}}_G}
    \end{bmatrix},
\end{align*}
and the Schur complement of $L$ with respect to the principal submatrix $L_{\overline{\mathcal{V}}_G \overline{\mathcal{V}}_G}$ denoted by $L/L_{\overline{\mathcal{V}}_G \overline{\mathcal{V}}_G}$ is defined as
\begin{align*}
    L/L_{\overline{\mathcal{V}}_G \overline{\mathcal{V}}_G} := L_{\mathcal{V}_G \mathcal{V}_G} - L_{\overline{\mathcal{V}}_G\mathcal{V}_G} L_{\overline{\mathcal{V}}_G\overline{\mathcal{V}}_G}^{-1}L_{\mathcal{V}_G \overline{\mathcal{V}}_G}.
\end{align*}
We denote the Kron-reduced graph with boundary nodes $\mathcal{V}_G$ as $\mathcal{G}_{\text{red}}$, the corresponding reduced Laplacian matrix $L_{\text{red}}$ is the Schur complement of $L$ with respect to $L_{\overline{\mathcal{V}}_G \overline{\mathcal{V}}_G}$. Fig.~\ref{fig:simple_kron_reduction} shows the Kron reduction operation on the graph $\mathcal{G}$ with boundary nodes $\mathcal{V}_G = \{1,3,7\}$, highlighted in red.
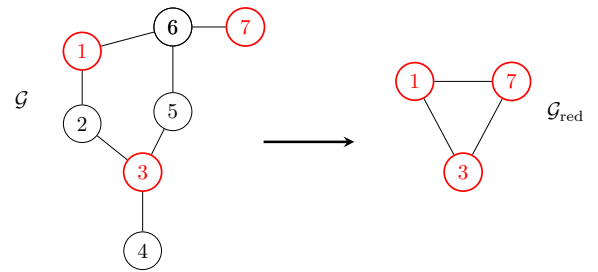
\begin{figure}[h!]
  \centering
  \scalebox{0.8}{
\begin{tikzpicture}[>=stealth, node distance=1.6cm]

\node[circle, draw, red, thick] (1) at (0,2) {1};
\node[circle, draw]            (6) at (1.5,2.4) {6};
\node[circle, draw]            (2) at (1.5,1) {5};
\node[circle, draw]            (5) at (0,0.8) {2};
\node[circle, draw, red, thick] (3) at (1,0) {3};
\node[circle, draw]            (4) at (1,-1.3) {4};
\node[circle, draw] (6) at (1.5,2.4) {6};
\node[circle, draw, red, thick] (7) at (2.7,2.4) {7};

\draw[-] (1) -- (6);
\draw[-] (6) -- (2);
\draw[-] (2) -- (3);
\draw[-] (5) -- (3);
\draw[-] (1) -- (5);
\draw[-] (4) -- (3);
\draw[-] (6) -- (7);

\node at (-1.0, 1.2) {$\mathcal G$};

\draw[->, very thick] (3,0.5) -- (4.5,0.5);

\node[circle, draw, red, thick] (r1) at (5.5,1.5) {1};
\node[circle, draw, red, thick] (r2) at (7.1,1.5) {7};
\node[circle, draw, red, thick] (r3) at (6.30,0) {3};

\draw[-] (r1) -- (r2);
\draw[-] (r2) -- (r3);
\draw[-] (r1) -- (r3);

\node at (8, 1) {$\mathcal G_{\mathrm{red}}$};

\end{tikzpicture}}
  \caption{Kron reduction of $\mathcal{G}$ with vertex set $\mathcal{V} = \{1, \ldots, 7\}$, boundary nodes $\mathcal{V}_G = \{1, 3, 7\}$, and interior nodes $\mathcal{V}_{\overline{\mathcal{V}}_G} = \{2, 4, 5, 6\}$.}
  \label{fig:simple_kron_reduction}
\end{figure}

Suppose the edge weights $a_{ij} \geq 0$ on the graph represent the conductances (i.e., admittance for power networks) of the branches $\{i,j\} \in \mathcal{E}$. The vector of node voltages $v$ and current injections $J$ are related by $Lv = J$ where $L$, the conductance matrix of the network, is the graph Laplacian matrix. The effective resistance between nodes $i$ and $j$ denoted by $r_{ij}^{\text{eff}}$ is defined as the potential difference between node $i$ and $j$ when a unit current is injected at node $i$ and withdrawn from node $j$, that is,
\begin{align*}
    r_{ij}^{\text{eff}} &= (e_i - e_j)^{\intercal} v, \qquad \text{where } v = L^{\dagger}J = L^{\dagger}(e_i - e_j)\\
    &= (e_i - e_j)^{\intercal}L^{\dagger}(e_i - e_j).
\end{align*}

The effective resistance defines a distance metric on the graph $\mathcal{G}$ and the distance between nodes $i$ and $j$ on the graph, denoted by $d(i,j) = (r_{ij}^{\text{eff}})^{1/2}$ \cite{ghosh_minimizing_2008}. We denote the matrix of effective resistance by $R \in \mathbb{R}^{n \times n}$ and the total effective resistance, a measure of the size of the graph $\mathcal{G}$, is defined as
\begin{align*}
    R_{\text{tot}}(\mathcal{G}) = \frac{1}{2} \sum_{i, j = 1}^n r_{ij}^{\text{eff}}(\mathcal{G}) = \frac{1}{2} \sum_{i, j \in \mathcal{V}} r_{ij}^{\text{eff}}(\mathcal{G}),
\end{align*}
where $\mathcal{V} = \{1, \ldots, n\}$. $R_{\text{tot}} (\mathcal{G})$ is also called the \textit{Kirchhoff Index} of the graph $\mathcal{G}$.

\subsection{Power System Model}
\label{subsec:power_system_model}
We consider a lossless $AC$ power network, modeled as an undirected, connected graph with nodes $\mathcal{V} = \{1, \ldots, n\}$, and branches $\mathcal{E}$ representing the set of transmission lines. The network is composed of generator nodes $\mathcal{V}_G$, and load nodes $\mathcal{V}_L$. For each node $i$, we denote the per unit bus voltage magnitude by $|V_i| \geq 0$ and the voltage angle (in radians) by $\delta_i \in \mathbb{S}^1$, where $\mathbb{S}^1$ is the unit circle. Given the 
symmetric nodal admittance matrix $Y \in \mathbb{C}^{n \times n}$, the network dynamics can be modeled by the differential algebraic equations (DAE),
\begin{align}
    m_i \ddot{\delta}_i + d_i \dot{\delta}_i &= p_i - \sum_{j = 1}^na_{ij}\sin(\delta_i - \delta_j) \quad \forall i \in \mathcal{V}_G \label{eqn:generator_model}\\
    0 &= p_i - \sum_{j = 1}^na_{ij}\sin(\delta_i - \delta_j) \quad \forall i \in \mathcal{V}_L
    \label{eqn:load_model}
\end{align}
where $a_{ij} 
= |V_i| |V_j| \Im{(Y_{ij})} 
\geq 0$, and $a_{ij}\sin(\delta_i - \delta_j)$ is the real power flow on line $\{i, j\} \in \mathcal{E}$ of susceptance $\Im{(Y_{ij})} = b_{ij}$, $p_i$ is the net power injection at node $i$, $d_i > 0$ and $m_i > 0$ are the damping and inertia coefficients of the $i$th generator, respectively. The second-order differential equation~(\ref{eqn:generator_model}) is the classic swing equation for generator $i$ and the algebraic equation~(\ref{eqn:load_model}) enforces power balance constraint at load node $i$. We make the following typical \cite{paganini_global_2017, paganini_global_2020, guo_graph_2018, nagpal_designing_2022} simplifying assumptions about the network:
\begin{enumerate}[label=\roman*.]
    \item The network is lossless, that is, for each transmission line $\{i, j\} \in \mathcal{E}$, the resistance $r_{ij} = 0$.
    \item Node voltage magnitudes $|V_i| \approx 1.0$ p.u.
    \item The effect of reactive power on frequency dynamics can be ignored.
    \item The generator buses have identical synchronous generators (that is, $d_i = d$, and $m_i = m$ for all $i \in \mathcal{V}_G$).
\end{enumerate}

\section{Problem Formulation}
\label{sec:problem_formulation}
Consider a linearization of the power system model (\ref{eqn:generator_model}) - (\ref{eqn:load_model}) around a synchronous solution $(\delta^*, \mathbf{0})$ the network dynamics can be written in compact form as
\begin{align}
    M \ddot{\Delta \delta} + D \dot{\Delta \delta} = - L \Delta \delta,
    \label{eqn:linearized_dynamics}
\end{align}
where $\Delta \delta$ is the deviation of the voltage angles $\delta$ from $\delta^*$, $M = \text{blkdiag}(M_{\mathcal{V}_G}, \mathbf{0})$, $D = \text{blkdiag}(D_{\mathcal{V}_G}, \mathbf{0})$, $M_{\mathcal{V}_G} = \text{diag}(\{m_i\}_{i \in \mathcal{V}_G})$, $D_{\mathcal{V}_G} = \text{diag}(\{d_i\}_{i \in \mathcal{V}_G})$ and $L$ is the Laplacian matrix of the graph defined on the vertex set $\mathcal{V}$ whose edge weights are $a_{ij} = \Im(Y_{ij})\cos(\delta_i^* - \delta_j^*)$. If we assume that the angle differences are small (that is, $\delta_i^* - \delta_j^* \ll 1$ rad), $\cos(\delta_i^* - \delta_j^*) \approx 1.0$ for all $\{i, j\} \in \mathcal{E}$, the edge weights $a_{ij} \approx b_{ij}$, the susceptance of the line $\{i,j\}$. Eliminating the algebraic constraints in (\ref{eqn:linearized_dynamics}) associated with the load nodes $\mathcal{V}_L$ using Kron reduction, we obtain the following differential equation
\begin{align}
    M_{\mathcal{V}_G} \ddot{\Delta \delta}_{\mathcal{V}_G} + D_{\mathcal{V}_G} \dot{\Delta \delta}_{\mathcal{V}_G} = -L_{\text{red}}\Delta\delta_{\mathcal{V}_G},
    \label{eqn:reduced_linerized_model}
\end{align}
$L_{\text{red}}$ is the reduced Laplacian matrix obtained by taking the Schur complement of $L$ with respect to the principal submatrix defined by the set of load buses $\mathcal{V}_L$. The solutions to the model (3) can be shown to be locally equivalent to the reduced model (4) for almost all operating states \cite{hiskens_energy_1989}. If we define $\dot{\Delta \delta}_{\mathcal{V}_G} = \omega$, the reduced network can be represented by the block diagram Fig.~\ref{fig:coupled_oscillator_block_diagram}, similar to \cite{paganini_global_2017, paganini_global_2020}.
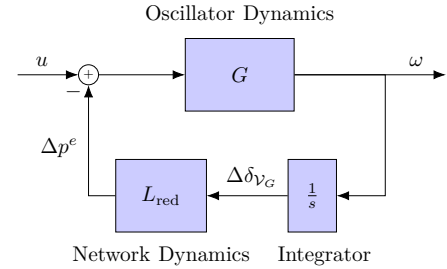
\begin{figure}[ht!]
  \centering
  \scalebox{0.8}{\begin{tikzpicture}[auto, >={Latex[length=2mm,width=1.5mm]}]

    \node[draw, minimum width=1.8cm, minimum height=1.2cm, fill=blue!20] (G) at (2.5,0) {$G$};
    \node at (2.5, 1) {\textbf{Oscillator Dynamics}};

    \node[draw, minimum width=1.5cm, minimum height=1.2cm, fill=blue!20] (Int) at (1.2,-2) {$L_{\text{red}}$};
    \node at (1.2, -3) {\textbf{Network Dynamics}};
    
    \node[draw, minimum width=0.8cm, minimum height=1.2cm, fill=blue!20] (L)   at (3.7,-2) {$\frac{1}{s}$};
    \node at (3.9, -3) {\textbf{Integrator}};

    \node[circle, draw, inner sep=1pt] (sum) at (0,0) {$\scriptscriptstyle +$};
    \node at (-0.25,-0.28) {$-$};

    \node[above left=15pt of sum.west, yshift=-10pt] {$u$};
    \node[above right=15pt of G.east, xshift=25pt, xshift=15pt, yshift=-10pt] {$ \omega$};
    
    \node[below left=2pt of sum.south, yshift=-20pt] {$\Delta p^{e}$};
    
    \node[above right=15pt of Int.east, xshift=-5pt, yshift=-10pt] {$\Delta \delta_{\mathcal{V}_G}$};

    \draw[->] (sum) -- (G);
    \draw[->] (G.east) -- ++(1.5,0) |- (L.east);
    \draw[->] (L.west) -- (Int.east);
    \draw[->] (Int.west) -| (sum.south);
    \draw[<-] (sum.west) -- ++(-1,0);
    \draw[->] (G.east) -- ++(2.5,0);

\end{tikzpicture}}
  \caption{Block diagram of second-order oscillators coupled through the network.}
  \label{fig:coupled_oscillator_block_diagram}
\end{figure}
The model (\ref{eqn:reduced_linerized_model}) is a network of second-order oscillators coupled through the graph $\mathcal{G}_{\text{red}}$ whose Laplacian matrix is $L_{\text{red}} \in \mathbb{R}^{k \times k}$, where $|\mathcal{V}_G| = k$ is the number of generator buses in the network, $u$ is an exogenous input (disturbance) to the network, $\omega$ is a vector of angular velocities of the generators, $\Delta p^e(s) = \frac{1}{s}L_{\text{red}} \omega(s)$ is a vector of network fluctuations at the generator nodes, $G(s) = \text{diag}(\{g_i(s)\}_{i \in \mathcal{V}_G})$ is a block diagonal matrix of the transfer functions for the generators, for each $i$,
\begin{align*}
    g_i(s) = \frac{\omega_i(s)}{u_i(s) - \Delta p_i^e(s)} = \frac{1}{m_is + d_i}.
\end{align*}
The matrix $L_{\text{red}}$ is symmetric and is diagonalizable by an orthonormal basis of eigenvectors $V \in\mathbb{R}^{k \times k}$ so that $L_{\text{red}} = V \Lambda V^{\intercal}$,
where $\Lambda = \text{diag}(\{\lambda_{i}\}_{i = 1}^k)$, and $0 = \lambda_1 < \lambda_2 \leq \cdots \leq \lambda_k$ are the real eigenvalues. We remark that by Theorem 5 of \cite{smith_interlacing_1992}, the eigenvalues of $L_{\text{red}}$ interlace those of $L$.
\begin{figure}[ht!]
  \centering
  \scalebox{0.5}{%
    \begin{tikzpicture}[auto, >={Latex[length=2mm,width=1.5mm]}]

    \node[draw, minimum width=1.8cm, minimum height=1.2cm, fill=blue!20] (G) at (3,0) {$G$};

    \node[draw, minimum width=1.2cm, minimum height=1.2cm, fill=blue!20] (Int) at (3,-2) {$\frac{1}{s} \Lambda$};

    \node[draw, minimum width=1cm, minimum height=1.2cm, fill=blue!20] (VT1)   at (-3,0) {$V^{T}$};
    \node[draw, minimum width=1cm, minimum height=1.2cm, fill=blue!20] (V1)   at (0.2,0) {$V$};

    \node[draw, minimum width=1cm, minimum height=1.2cm, fill=blue!20] (VT2)   at (6,0) {$V^{\intercal}$};

    \node[circle, draw, inner sep=1pt] (sum) at (-1.5,0) {$\scriptscriptstyle +$};
    \node at (-1.75,-0.28) {$-$};

    \node[draw, minimum width=1cm, minimum height=1.2cm, fill=blue!20] (V2)   at (9,0) {$V$};

    \node[above left=15pt of VT1.west, yshift=-10pt] {$u$};


    \node[above right=15pt of V2.east, xshift=-5pt, yshift=-10pt] {$\omega$};

    \draw[->] (V1.east) -- (G.west);
    \draw[<-] (VT1.west) -- ++(-1,0);
    \draw[->] (VT1.east) -- (sum.west);
    \draw[->] (sum.east) -- (V1.west);
    \draw[->] (G.east) -- (VT2.west);
    \draw[->] (VT2.east) -- (V2.west);
    \draw[->] (VT2.east) -- ++(1,0) |- (Int.east);
    \draw[->] (Int.west) -| (sum.south);
    \draw[->] (V2.east) -- ++(1,0);

\end{tikzpicture}}
  \caption{Spectral transformation of feedback loop.}
  \label{fig:spectral_transformation}
\end{figure}
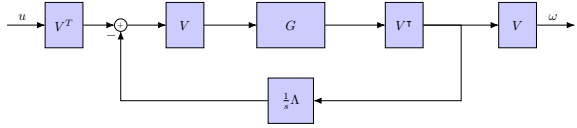
A spectral transformation of the feedback loop as in Fig.~\ref{fig:spectral_transformation} allows for a decomposition and reduction of the model (\ref{eqn:reduced_linerized_model}). 
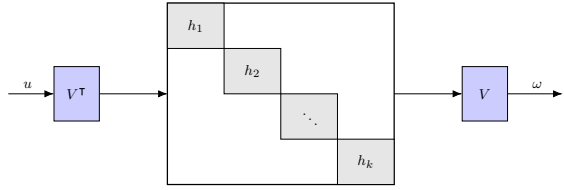
\begin{figure}[h!]
  \centering
  \scalebox{0.6}{%
    \begin{tikzpicture}[every node/.style={font=\small}, block/.style={draw, fill=gray!20}, auto, >={Latex[length=2mm,width=1.5mm]}]

    \def\H{1}          
    \def\W{1.25}       
    \def\N{4}          
    
    \draw[thick] (0,0) rectangle (5,4);
    
    
    \foreach \i in {0,...,3} {
        \draw[block]
            (\i*\W, {4-(\i+1)*\H})
            rectangle
            ({(\i+1)*\W}, {4-\i*\H});
    }
    
    \node at (0.625,3.5) {$h_1$};
    \node at (1.875,2.5) {$h_2$};
    \node at (3.15,1.55) {$\ddots$};
    \node at (4.375,0.5) {$h_k$};
    \node[draw, minimum width=1cm, minimum height=1.2cm, fill=blue!20] (V)   at (-2,2) {$V^{\intercal}$};
    \node[draw, minimum width=1cm, minimum height=1.2cm, fill=blue!20] (VT)   at (7,2) {$V$};
    \node[above left=0pt of V.west, xshift = -10pt, yshift=-0pt] {$u$};

    \node[above left=0pt of VT.east, xshift = 25pt, yshift=-0pt] {$\omega$};
    \draw[->] (5, 2) -- (VT.west);
    \draw[->] (V.east) -- (0, 2);
    \draw[->] (VT.east) -- ++(1.2, 0);
    \draw[<-] (V.west) -- ++(-1.0, 0);
\end{tikzpicture}}
  \caption{Block diagram of the closed loop transfer function from disturbances $u$ to the angular velocities $\omega$.}
  \label{fig:u_to_w_transfer_function}
\end{figure}
Eliminating the feedback path, we obtain the closed-loop transfer function with block diagram shown in Fig.~\ref{fig:u_to_w_transfer_function}. The closed loop transfer function from disturbance $u(s)$ to angular frequency $\omega(s)$
\begin{align*}
    &T_{\omega u}(s) = \frac{\omega(s)}{u(s)} = VH(s)V^{\intercal}, \quad H(s) = \text{diag}(h_i(s)).\\
    &\text{where} \quad h_i(s) = \frac{s g(s)}{s + \lambda_i g(s)}\quad i = 1, \ldots, k.
\end{align*}
For a step input defined as
\begin{align*}
    u(t) = \begin{cases}
        0, & t < 0,\\
        u_{0}\in \mathbb{R}^k, & t\geq 0,
    \end{cases}
\end{align*}
the angular frequency in time domain can be decomposed into
\begin{align*}
    \omega(t) = \overline{\omega}(t) \mathbf{1} + \widetilde{\omega}(t),
\end{align*}
where $\overline{\omega}(t)$ is a scalar-valued function of the system frequency {(the center of inertia)} that only depends on the generator parameters ($d$ and $m$) and the disturbance $u_0$, with no dependence on the network at all; $\widetilde{\omega}(t)$ is a vector-valued function whose entries are the deviations of the angular frequencies at the generator nodes from the system synchronous response. The term $\widetilde{\omega}(t)$ depends on the power network through the non-zero eigenvalues and the corresponding eigenvector of $L_{\text{red}}$, see Appendix~\ref{Appendix_1}. This decoupling clearly shows that there is a component of the overall angular frequency response that cannot be reduced by modifying the network in any way using the linearized model (\ref{eqn:reduced_linerized_model}) since the component $\overline{\omega}(t)\mathbf{1}$ evolves in $\text{Ker}(L_{\text{red}}) = \text{Span}(\mathbf{1})$, the Nullspace of $L_{\text{red}}$.

The $\mathcal{L}_2$ norm of the transient term (the cost of synchronization)
\begin{align*}
    \|\widetilde{\omega}\|_{2}^2
    = \int_{0}^{\infty} \widetilde{\omega}(t)^{\intercal} \widetilde{\omega}(t) dt = \frac{1}{2d}\sum_{i = 2}^{k} \frac{(v_{i}^{\intercal}u_0)^2}{\lambda_{i}},
\end{align*}
where $(\lambda_i, v_i)$ is the $i$th eigenpair of $L_{\text{red}}$, see Appendix~\ref{Appendix_1} for details. The term $v_i^{\intercal}u_0$ is the component of the projection of the disturbance $u_0$ in the direction of the $i$th eigenvector $v_i$. For an exogenous input $u_0 \sim \mathcal{N}(0, \sigma^2 I)$ the expected $\mathcal{L}_2$ norm of the transient term,
\begin{align*}
    \mathbb{E}_{u_0}[\|\widetilde{\omega}\|_2^2] 
    = \frac{\sigma^2}{2d} \frac{1}{k} R_{\text{tot}} (\mathcal{G}_{\text{red}}),
\end{align*}
where $R_{\text{tot}}(\mathcal{G}_{\text{red}})$ is the total effective resistance of the Kron-reduced graph $\mathcal{G}_{\text{red}}$, $d$ is the damping coefficient of the generators, and $k$ is the number of generator nodes on the graph $\mathcal{G}$.

\section{Optimization Problems}
\label{sec:methodology}
In this section, we show that under the stated assumptions on the power system, the expected $\mathcal{L}_2$ norm of the transient component $\widetilde{\omega}(t)$ for an exogenous input $u_0 \sim \mathcal{N}(0, \sigma^2 I)$ is a convex function of the edge weights of the graph $\mathcal{G}$ (the susceptances of the transmission lines).

\subsection{The SDP Minimizing Frequency Transients}
\label{subsec:sdp_formulation}
We formulate a convex optimization problem to minimize this expectation by optimally allocating a budget $\alpha \geq 0$ across a set of candidate edges $\mathcal{E}_c \subseteq \mathcal{V} \times \mathcal{V}$. Essentially, we are interested in modifying the coupling matrix $L_{\text{red}}$ (in Fig.~\ref{fig:coupled_oscillator_block_diagram}) between the second-order oscillators through appropriate updates to the
edge weights $a_{ij}$ of the graph $\mathcal{G}$ in order to minimize $\mathbb{E}_{u_0} \left[\|\widetilde{w}\|_2^2\right]$.

We introduce the following preliminary propositions to facilitate the formulation of this problem.
\begin{proposition}
    For a connected undirected graph $\mathcal{G}$ with an irreducible Laplacian matrix $L$, the matrix $L_{\text{red}}$ is an irreducible Laplacian for any set of boundary nodes $\mathcal{V}_G \subset \mathcal{V}$.
    \label{L_red_is_laplacian}
\end{proposition}
\begin{proof}
    The proof follows from Theorem 4.3 \cite{ayazifar_graph_nodate}, since $L$ is a singular $M$-matrix and $L_{\text{red}}$ is the Schur complement of $L$ with respect to the principal submatrix defined by $\overline{\mathcal{V}}_G \subset \mathcal{V}$.
\end{proof}
\noindent
Proposition~\ref{L_red_is_laplacian} ensures that the Schur complement is well defined for any connected graph and that the resulting matrix $L_{\text{red}}$ is indeed an irreducible Laplacian matrix. 
\begin{lemma}
    The pairwise effective resistance $r_{ij}^{\text{eff}}$ is invariant under Kron reduction, that is, $r_{ij}^{\text{eff}}(\mathcal{G}_{\text{red}}) = r_{ij}^{\text{eff}}(\mathcal{G})$ for all $i,j \in \mathcal{V}_G \subset \mathcal{V}$.
    \label{invariance_reff_ij}
\end{lemma}
\begin{proof}
    See Appendix~\ref{appendix:invariance_of_r_ij_proof}.
\end{proof}
This invariance under Kron reduction allows us to dispense with an explicit characterization of the effective resistance in terms of the underlying graph when the graph under consideration is clear from context. In the remainder of this work, we use $r_{ij}^{\text{eff}}$ to denote the effective resistance between generator nodes $i$ and $j$. A consequence of the invariance of $r_{ij}^{\text{eff}}$, is that the total effective resistance satisfies $R_{\text{tot}}(\mathcal{G}_{\text{red}}) \leq R_{\text{tot}}(\mathcal{G})$.
\begin{lemma}
    The effective resistance $r_{ij}^{\text{eff}} 
    = (e_i - e_j)^{\intercal}L_{\text{reg}, \beta}^{-1}(e_i - e_j)$ for all $\beta \neq 0$.
    \label{r_ij_equality_pinvL_invL}
\end{lemma}
\begin{proof}
    The regularized Laplacian is defined as, $L_{\text{reg}, \beta} =\beta v_1 v_1^{\intercal} +  \sum_{i = 2}^n \lambda_i v_i v_i^{\intercal}$, and for all $\beta \neq 0$,
    \begin{align*}
        L_{\text{reg}, \beta}^{-1} &= \frac{1}{\beta} v_1 v_1^{\intercal} +  \sum_{i = 2}^n \frac{1}{\lambda_i} v_i v_i^{\intercal} = \frac{1}{\beta n} \mathbf{1} \mathbf{1}^{\intercal} +  \sum_{i = 2}^n \frac{1}{\lambda_i} v_i v_i^{\intercal}.
    \end{align*}
    The quadratic form,
    \begin{align*}
        (e_i - e_j)^{\intercal}L_{\text{reg}, \beta}^{-1}(e_i - e_j) &= \underbrace{\frac{1}{\beta n} (e_i - e_j)^{\intercal} \mathbf{1} \mathbf{1}^{\intercal} (e_i - e_j)}_{= 0}\\
        & + (e_i - e_j)^{\intercal}\sum_{i = 2}^n \frac{1}{\lambda_i} v_i v_i^{\intercal} (e_i - e_j)\\
        &= (e_i - e_j)^{\intercal} L^{\dagger} (e_i - e_j) = r_{ij}^{\text{eff}},
    \end{align*}
    by the definition of effective resistance Section~\ref{subsec:kron_reduction_and_effective_resistance}.
\end{proof}
\begin{theorem}
    For a connected undirected graph $\mathcal{G}$ with Laplacian matrix $L$, the total effective resistance of the Kron-reduced graph $R_{\text{tot}}(\mathcal{G}_{\text{red}})$ with boundary nodes $\mathcal{V}_G \subset \mathcal{V}$ is a linear function of $L_{\text{reg}, \beta}^{-1}$.
    \label{r_tot_linear_in_Lreg_inv}
\end{theorem}

\begin{proof}
    Given a graph $\mathcal{G} = (\mathcal{V}, \mathcal{E}, A)$, let $\mathcal{V}_{G} \subset \mathcal{V}$ and let $\mathbf{1}_{\mathcal{V}_G}$ denote the indicator vector for the set $\mathcal{V}_G$ defined as
    \begin{align*}
        [\mathbf{1}_{\mathcal{V}_G}]_i = \begin{cases}
            1 & \text{if }i \in \mathcal{V}_G\\
            0 & \text{otherwise},
        \end{cases}
    \end{align*}
    also define an orthonormal matrix $E \in \mathbb{R}^{n \times k}$ whose columns are the standard basis vectors $e_i$ for each $i \in \mathcal{V}_G$.
    By definition, the total effective resistance of $\mathcal{G}_{\text{red}}$ with boundary nodes $\mathcal{V}_G$,
    \begin{align*}
        R_{\text{tot}}(\mathcal{G}_{\text{red}})
        = \frac{1}{2}\sum_{i,j \in \mathcal{V}_G} r_{ij}^{\text{eff}}.
    \end{align*}
    Using Lemma~\ref{r_ij_equality_pinvL_invL},
    \begin{align*}
        R_{\text{tot}}(\mathcal{G}_{\text{red}}) = \frac{1}{2}\sum_{i,j \in \mathcal{V}_G} (e_i - e_j)^{\intercal}L_{\text{reg}, \beta}^{-1}(e_i - e_j).
    \end{align*}
    Define $Y = L_{\text{reg}, \beta}^{-1}$,
    \begin{align*}
        R_{\text{tot}}(\mathcal{G}_{\text{red}}) &= \frac{1}{2} \sum_{i, j \in \mathcal{V}_G} (Y_{ii} + Y_{jj} - 2Y_{ij})\\
        &= k \sum_{i \in \mathcal{V}_G} Y_{ii} - \sum_{i,j \in \mathcal{V}_G} Y_{ij}\\
        &= k \, \text{trace}(Y_{\mathcal{V}_G}) - \mathbf{1}_{\mathcal{V}_G}^{\intercal}Y\mathbf{1}_{\mathcal{V}_G},
    \end{align*}
    where $Y_{\mathcal{V}_G}$ is the principal submatrix of $Y$ defined by the set $\mathcal{V}_G$, it is straightforward to verify that $Y_{\mathcal{V}_G} = EYE^{\intercal}$. Using the cyclic property and linearity of trace,
    \begin{align*}
        R_{\text{tot}}(\mathcal{G}_{\text{red}}) &= k \, \text{trace}(EYE^{\intercal}) - \text{trace} ({\mathbf{1}_{\mathcal{V}_G}^{\intercal}Y\mathbf{1}_{\mathcal{V}_G}})\\ 
        &= k \, \text{trace}\left(C^{\intercal} Y\right),
    \end{align*}
    where $C = EE^{\intercal} - \frac{1}{k}\mathbf{1}_{\mathcal{V}_G} \mathbf{1}_{\mathcal{V}_G}^{\intercal}$. Notice that for a fixed set of boundary nodes $\mathcal{V}_G \subset \mathcal{V}$, the matrix $C$ is constant, hence the total effective resistance on the Kron-reduced graph $R_{\text{tot}}(\mathcal{G}_{\text{red}})$ and $\text{trace}(L_{\text{red}}^{\dagger})$ is a linear function of the regularized inverse $L_{\text{reg}, \beta}^{-1}$.
\end{proof}
As a direct consequence of Theorem~\ref{r_tot_linear_in_Lreg_inv}, we obtain the following corollary.
\begin{corollary}
    For $u_0 \sim \mathcal{N}(0, \sigma^2 I)$, the expected $\mathcal{L}_2$ norm of the transient component of the angular frequency $\widetilde{w}(t)$ at the generator nodes is a linear function of $L_{\text{reg}, \beta}^{-1}$; specifically,
    \begin{align*}
        \mathbb{E}_{u_0} \left[\|\widetilde{w}\|_2^2\right] 
        = \frac{\sigma^2}{2d}\, \text{trace}(C^{\intercal}L_{\text{reg}, \beta}^{-1}),
    \end{align*}
    where $C$ is defined as in Theorem~\ref{r_tot_linear_in_Lreg_inv}.
    \label{corollary_1}
\end{corollary}

For fixed $d$, $k$, and $\sigma$, $\mathbb{E}_{u_0}\left[\|\widetilde{w}\|_2^2\right]$ can be minimized by minimizing the total effective resistance of $\mathcal{G}_{\text{red}}$ or the trace of $C^{\intercal}L_{\text{reg}, \beta}^{-1}$. We note that while the norm of the transient component derived in Section~\ref{sec:problem_formulation} is defined on the Kron-reduced graph $\mathcal{G}_{\text{red}}$, directly minimizing this term over this reduced network does not map to the actual edges on the network $\mathcal{G}$ in any meaningful way; topological/structural information about the original network is lost following network reduction \cite{ayazifar_graph_nodate}, \cite{dorfler_synchronization_2010}.
\begin{lemma}
    For a connected graph $\mathcal{G}$, the regularized laplacian $L_{\text{reg}, \beta}$ is positive definite for all $\beta > 0$.
    \label{L_reg_positive_definite}
\end{lemma}
\begin{proof}
    The graph $\mathcal{G}$ is connected implies that the algebraic connectivity, $\lambda_2 > 0$
    and the eigenvalues of $L$ can be ordered as $0 = \lambda_1 < \lambda_2 \leq \cdots \leq \lambda_n$. By spectral theorem of symmetric matrices, $L$ has spectral decomposition
    \begin{align*}
        L 
        = \sum_{i = 1}^n\lambda_i v_i v_i^{\intercal} = \lambda_1 v_1 v_1^{\intercal} +  \sum_{i = 2}^n \lambda_i v_i v_i^{\intercal},\\
        \text{where } v_1 = \frac{1}{\sqrt{n}}\mathbf{1}.
    \end{align*}
    For $\beta \neq 0$, the regularized Laplacian is defined as
    \begin{align*}
        L_{\text{reg}, \beta} &= L + \beta\frac{1}{n} \mathbf{1}\mathbf{1}^{\intercal} 
        = (\lambda_1 + \beta) v_1 v_1^{\intercal} +  \sum_{i = 2}^n \lambda_i v_i v_i^{\intercal}\\
        &= \beta v_1 v_1^{\intercal} +  \sum_{i = 2}^n \lambda_i v_i v_i^{\intercal}, \quad \text{since }\lambda_1 = 0.
    \end{align*}
    Notice that $L_{\text{reg}, \beta}$ is symmetric, and $\text{Spec}(L_{\text{reg}, \beta}) = \{\beta, \lambda_2, \cdots, \lambda_n\}$ is the multi-set of real eigenvalues, since $0<\lambda_2\leq\ldots \leq \lambda_n$, it follows that $L_{\text{reg}, \beta}$ is positive definite for any $\beta > 0$.
\end{proof}
\begin{proposition}
    The matrix $C$ is an orthogonal projector on to $\left\{x \;\left| \; \mathbf{1}_{\mathcal{V}_G}^{\intercal} x = 0,\right.\; x_j = 0 \;\forall j \notin \mathcal{V}_G\right\}$.
    \label{C_is_projector}
\end{proposition}
\begin{proof}
    see Appendix~\ref{appendix:c_is_projector}.
\end{proof}
\noindent
This implies that $C$ is symmetric positive semidefinite, using a congruence transformation of $Y$ and the cyclic property of trace,
\begin{align}
    \text{trace}(C^{\intercal}Y) = \text{trace}(C^{1/2}YC^{1/2}) \geq 0.
    \label{eqn:trace_c_L_reg}
\end{align}
for any $Y \succeq 0$.

For the graph $\mathcal{G}$, let $\mathcal{E}_c \subseteq \mathcal{V} \times \mathcal{V}$ denote the set of candidate edges considered in our formulation, and let $B_c \in \mathbb{R}^{n \times |\mathcal{E}_c|}$ denote the corresponding incidence matrix. We define a design vector $x \in \mathbb{R}^{|\mathcal{E}_c|}_{\geq 0}$, whose entries are the updates to the weights of the candidate edges in $\mathcal{E}_c$. The Laplacian matrix $L(x) = BWB^{\intercal} + B_c XB_c^{\intercal},$ where $W = \text{diag}(\{a_{ij}\}_{\{i,j\} \in \mathcal{E}})$ and $X = \text{diag}(x)$ is the optimization variable for our problem and we denote the corresponding graph by $\mathcal{G}(x) = (\mathcal{V}, \mathcal{E}, A(x))$. We consider the set of candidate edges $\mathcal{E}_c = \mathcal{E}$ and a budget of edge weights $\alpha \ge 0$. Accordingly, the problem formulation seeks an optimal allocation of weights over the existing edge set $\mathcal{E}$, without altering the underlying network topology. We note, however, that the formulations developed in this section, apply more generally to any set of candidate edges $\mathcal{E}_c$, not necessarily equal to $\mathcal{E}$. It is easy to verify that the matrix variable $L(x)$ is an irreducible positive semidefinite matrix for any edge set $\mathcal{E}_c \subseteq \mathcal{V} \times \mathcal{V}$ and $x \geq 0$, if $L$ is irreducible positive semidefinite.

\begin{theorem}
    For any $\mathcal{E}_c \subseteq \mathcal{V} \times \mathcal{V}$ and $\mathcal{V}_G \subset \mathcal{V}$, $R_{\text{tot}}(\mathcal{G}_{\text{red}}(x))$ is a convex function of $x \in \mathbb{R}^{|\mathcal{E}_c|}_{\geq 0}$.
    \label{convexity_of_r_tot}
\end{theorem}

\begin{proof}
    For any $z\in \mathbb{R}^n$, the function $f(L(x)_{\text{reg}, \beta}) = z^{\intercal}L(x)_{\text{reg}, \beta}^{-1}z$, is a convex, non-increasing function of the matrix variable $L(x)_{\text{reg}, \beta}$ if $L(x)_{\text{reg}, \beta} \succ 0$.\\
    By definition, $L(x)_{\text{reg}, \beta} = BWB^{\intercal} + B_{c}XB_c^{\intercal} + \frac{1}{\beta} \mathbf{11}^{\intercal}$ is an affine function of $x$ and from Lemma~\ref{L_reg_positive_definite}, we know that for $\beta >0$ and $x \geq 0$, $L(x)_{\text{reg}, \beta} \succ 0$. It follows from the composition of a convex non-increasing function with an affine function that $f(L(x)_{\text{reg}, \beta})$ is a convex function of $x$ (for any fixed $z$); section 3.2.4\cite{boyd_convex_2023}.
    
    Hence, by Lemma~\ref{r_ij_equality_pinvL_invL} the effective resistance $r_{ij}^{\text{eff}}(\mathcal{G}_{\text{red}}(x)) = (e_i - e_j)^{\intercal}L(x)^{-1}_{\text{reg}, \beta}(e_i - e_j)$ is a convex function for $x \geq 0$ and any $i, j \in \mathcal{V}$. Convexity of $R_{\text{tot}}(\mathcal{G}_{\text{red}}(x))$ follows by its definition as a non-negative sum of convex functions.
\end{proof}

We remark that the convexity result in Theorem~\ref{convexity_of_r_tot} simply requires that the regularized Laplacian remains positive definite over some domain of $x$, that is, convexity holds if we can directly constrain the non-zero eigenvalues of $L(x)$. It follows from Theorem~\ref{convexity_of_r_tot} that the optimization problem
\begin{align*}
        P_0: \quad \operatorname*{minimize}_{x} \quad & \text{trace}(C^{\intercal}L(x)_{\text{reg}, \beta}^{-1})\\
        \text{subject to} \quad &\alpha \geq x^{\intercal}\mathbf{1} , \; x \geq 0, \; \beta>0
\end{align*}
is convex for any $\alpha \geq 0$.

From the perspective of distances on the graph, the solution to $P_0$ determines $x^{\star}$ in the scaled-simplex that minimizes the sum of the squared ``electrical distances," $d(i,j)^2$ between the generator nodes. Introducing an auxiliary matrix variable $Y$ and a constraint $Y \succeq L(x)_{\text{reg}, \beta}^{-1}$, we can reformulate $P_0$ as a standard semidefinite program in the matrix variables $X$ and $Y$ as follows. 
Applying the Schur complement property for positive definite matrices \cite{boyd_linear_1994},
\begin{align*}
    Y \succeq L(x)_{\text{reg}, \beta}^{-1} \iff \begin{bmatrix}
        L(x)_{\text{reg}, \beta} & I \\ I & Y
    \end{bmatrix} \succeq 0.
\end{align*}
Therefore solving $P_0$ is equivalent to solving the following semidefinite program
\begin{align*}
    P_1: \quad  \operatorname*{minimize}_{X, Y} \quad & \text{trace}(C^{\intercal}Y)\\
    \text{subject to} \quad & \begin{bmatrix}
        L(x)_{\text{reg}, \beta} & I \\ I & Y
    \end{bmatrix} \succeq 0\\
    &\alpha \geq \mathbf{1}^{\intercal}X\mathbf{1},\\
    &X \succeq 0, \; \beta > 0.
\end{align*}

\begin{proposition}
    Given two connected graphs $\mathcal{G}$ and $\tilde{\mathcal{G}}$ defined on an identical set of nodes with possibly varying edges and edge weights, if $\tilde{a}_{ij} \geq a_{ij}$ for all $i, j \in \{1, \ldots, n\}$ then $\tilde{r}_{ij}^{\text{eff}} \leq r_{ij}^{\text{eff}}$ for all $i, j \in \{1, \ldots, n\}$.
    \label{r_ij_rayleigh_monotonicity}
\end{proposition}

\begin{proof}
    Rayleigh monotonicity law, see Proposition 5.5 Florian \textit{et al.}
    \cite{dorfler_electrical_2018}.
\end{proof}

By Proposition~\ref{r_ij_rayleigh_monotonicity}, for any
$x, y \in \mathbb{R}^n_{\geq 0}$ satisfying
$y \succeq x$, it holds that
\begin{align*}
    r_{ij}^{\mathrm{eff}}(\mathcal{G}(y))
    \leq
    r_{ij}^{\mathrm{eff}}(\mathcal{G}(x)),
    \qquad
    \forall i,j \in \{1,\dots,n\}.
\end{align*}
Using Lemma~\ref{invariance_reff_ij}, the total effective resistance of
$\mathcal{G}_{\mathrm{red}}(x)$ is monotonically non-increasing
with respect to $x$ and $\alpha$. Consequently, the solution to
$P_0$ always lies on the boundary of the feasible set.

\subsection{Ensuring Steady-State Phase Cohesion}
\label{subsec:gamma_cohesiveness}

Beyond minimizing the norm of the transient component, 
in this section, we formulate an LMI that ensures that for a set of net power injections, the frequency-synchronized state $(\delta^*, \mathbf{0})$ on the graph $\mathcal{G}(x)$ is stable with steady-state angles that achieve $\gamma$-cohesion. The state $(\delta^*, \mathbf{0})$ is $\gamma$-cohesive if the geodesic distance between the steady-state phase angles satisfies $|\delta_i^* - \delta^*_j| \leq \gamma$ rad for every edge $\{i, j\} \in \mathcal{E}$. In power networks phase-cohesion is often required to ensure semi-global tracking of desired power injections\cite{zholbaryssov_distributed_2018}.

\begin{theorem}
    For $\gamma \in (0, \pi/2)$ and $\psi \geq 0$, if the intersection of the convex set defined by the LMI
    \begin{align*}
        L(x) \succeq \frac{1}{\sin(\gamma)} \|B\|_2\psi \Pi_n, \quad \Pi_n = I - \frac{1}{n}\mathbf{1}\mathbf{1}^{\intercal}
    \end{align*} and the scaled-simplex $\mathcal{X}_{\alpha}$ is non-empty, then there exists a unique, exponentially stable $\gamma$-cohesive synchronized solution on the optimized network $\mathcal{G}(x^{\star})$ for all power injections
    \begin{align*}
        p \in \mathcal{P}_{\psi} = \{p \in \mathbf{1}^{\perp} \; | \; \|p\|_{2} \leq \psi\}.
    \end{align*}
    \label{existence_of_gamma_cohesive_solution}
\end{theorem}

\begin{proof}
    For a net power injection vector $p \in \mathbf{1}^{\perp}$, a unique exponentially stable $\gamma$-cohesive synchronized state $(\delta^*, \mathbf{0})$ exists if the algebraic connectivity of $\mathcal{G}(x)$ satisfies,
    \begin{align}
        \|B\|_2 \|p\|_2 \frac{1}{\sin(\gamma)} \leq \lambda_2,
        \label{eqn:synchronization_condition_2}
    \end{align}
    proof in Appendix~\ref{appendix:sufficient_sync_condition}.
    
    For all $x \in \mathbb{R}^{|\mathcal{E}|}$, the Laplacian $L(x)$ is symmetric and has spectral decomposition $L(x) = \sum_{i = 1}^n \lambda_i v_i v_i^{\intercal}$, where $\lambda_k = 0$ for some $k$ with the corresponding eigenvector $\mathbf{1}$. To satisfy the synchronization condition $(5)$, we constrain the non-zero eigenvalues of the matrix variable $L(x)$. Define a projector on $\mathbf{1}^{\perp}$ using the orthonormal set of eigenvectors $v_i$ as,
    \begin{align*}
        \Pi_{n} = \sum_{i = 2}^n v_i v_i^{\intercal} = I - \frac{1}{n}\mathbf{1}\mathbf{1}^{\intercal}.
    \end{align*} 
     Notice that $\Pi_n \succeq 0$ with eigenvalues $\text{Spec}(\Pi_n) = \{0, 1, 1\ldots, 1\}$. 
    So the inequality (\ref{eqn:synchronization_condition_2}) holds if and only if,
    \begin{align*}
        \lambda_2 \Pi_{n} &\succeq \frac{1}{\sin(\gamma)} \|B\|_2 \|p\|_2 \Pi_n. 
    \end{align*}
    It follows that if the LMI
    \begin{align*}
        L(x) &\succeq 
        \frac{1}{\sin(\gamma)} \|B\|_2 \|p\|_2 \Pi_n,
    \end{align*}
    is feasible, the $n-1$ non-zero eigenvalues of $L(x)$ satisfies the inequality (\ref{eqn:synchronization_condition_2}). Therefore, for all net power injection vectors $p \in \mathcal{P}_{\psi} = \{p \in \mathbf{1}^{\perp} \; | \; \|p\|_{2} \leq \psi\}$ on the graph $\mathcal{G}(x)$, whose Laplacian satisfies
    \begin{align}
        L(x) &\succeq \frac{\psi}{\sin(\gamma)} \|B\|_2  \Pi_n,
        \label{eqn:LMI_1}
    \end{align}
    the network dynamics (\ref{eqn:generator_model}) - (\ref{eqn:load_model}) has a unique and exponentially stable $\gamma$-cohesive synchronized state $(\delta^*, \mathbf{0})$ and if the intersection of the set defined by (\ref{eqn:LMI_1}) and $\mathcal{X}_{\alpha}$ is non-empty, the synchronization condition (\ref{eqn:synchronization_condition_2}) holds for $L(x^{\star})$. This completes the proof.
\end{proof}

We denote the set of $x \in \mathbb{R}^{|\mathcal{E}_c|}$ such that (\ref{eqn:LMI_1}) holds by $\mathcal{X}_{\text{sync}}$, and it is straightforward to verify that $\mathcal{X}_{\text{sync}}$ is a convex set \cite{caverly_lmi_2024}. Since the real power flow on the lossless branch $\{i,j\}$ is given by $p_{ij} = a_{ij}\sin(\delta_i - \delta_j)$, the inequality (\ref{eqn:synchronization_condition_2}) ensures that the edge weights $a_{ij}$ on the graph $\mathcal{G}(x)$ are sufficiently large to guarantee that the angle differences required for the net power injections $p \in \mathcal{P}_{\psi}$, is less than $\gamma$. 
We note that if $\psi > 0$ and $\gamma \in (0, \pi/2)$, the synchronization condition (\ref{eqn:LMI_1}) is sufficient to ensure that the Laplacian $L(x)$ remains irreducible (i.e., $\lambda_2  > 0$ and $\mathcal{G}(x)$ remains connected) when we relax the non-negativity constraint on the design vector $x$.

Solving the following SDP 
\begin{align*}
    P_2: \quad  \operatorname*{minimize}_{X, Y} \quad & \text{trace}(C^{\intercal}Y)\\
    \text{subject to} \quad & \begin{bmatrix}
        L(x)_{\text{reg}, \beta} & I \\ I & Y
    \end{bmatrix} \succeq 0\\
    &L(x) \succeq \frac{\psi}{\sin(\gamma)} \|B\|_2 \Pi_n\\
    &\alpha \geq \mathbf{1}^{\intercal}X\mathbf{1}, \; 
    X \succeq 0, \; 
    \beta > 0
\end{align*}
minimizes transients while ensuring that the steady-state angle cohesiveness requirement is met for the set of power injections $\mathcal{P}_{\psi}$ on the optimal graph $\mathcal{G}(x^{\star})$. While the optimization problem $P_0$ is always feasible for any budget $\alpha \geq 0$, the feasibility of $P_2$ strongly depends on the choice of parameters $\gamma$, and $\psi$. Specifically, for a fixed $\gamma$, the intersection of the simplex $\mathcal{X}_{\alpha}$ and $\mathcal{X}_{\text{sync}}$ shrinks as $\psi$ increases (that is, for larger sets $\mathcal{P}_{\psi}$, it becomes more difficult to find $x \in \mathcal{X}_{\alpha}$ that satisfy the $\gamma$-cohesiveness condition (\ref{eqn:LMI_1})). Consequently, for sufficiently large $\psi$, the problem $P_2$ can be infeasible.

Consider an example of the optimization problem $P_2$ defined for a simple path graph with generator nodes $\{1,2\}$ and load node $\{3\}$ shown in Fig.~\ref{fig:simple_path_graph}. We consider candidate edges $\mathcal{E}_c = \mathcal{E}$, so that the design vector for this problem is two-dimensional with $x_1$ and $x_2$ corresponding to the branches $\{1,2\}$ and $\{2,3\}$, respectively.

\begin{figure}[htpb]
  \centering
  \begin{tikzpicture}[>=stealth, node distance=1.6cm]

\node[circle, draw, red, thick] (n1) at (0,1.5) {1};
\node[circle, draw, red, thick] (n2) at (2.5,1.5) {2};
\node[circle, draw, black, thick] (n3) at (5.0,1.5) {3};

\draw[-, line width=0.5pt] (n1) -- node[midway, above] {2.0} (n2);
\draw[-, line width=2pt] (n2) -- node[midway, above] {5.0} (n3);

\end{tikzpicture}
  \caption{Weighted undirected path graph on $3$ nodes, with generator nodes $\mathcal{V}_G = \{1, 2\}$ and load nodes $\mathcal{V}_L = \{3\}$.}
  \label{fig:simple_path_graph}
\end{figure}
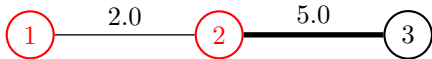

\begin{figure}[htpb]
    \centering
    \begin{minipage}[t]{0.32\linewidth}
        \centering
        \includegraphics[width=\linewidth]{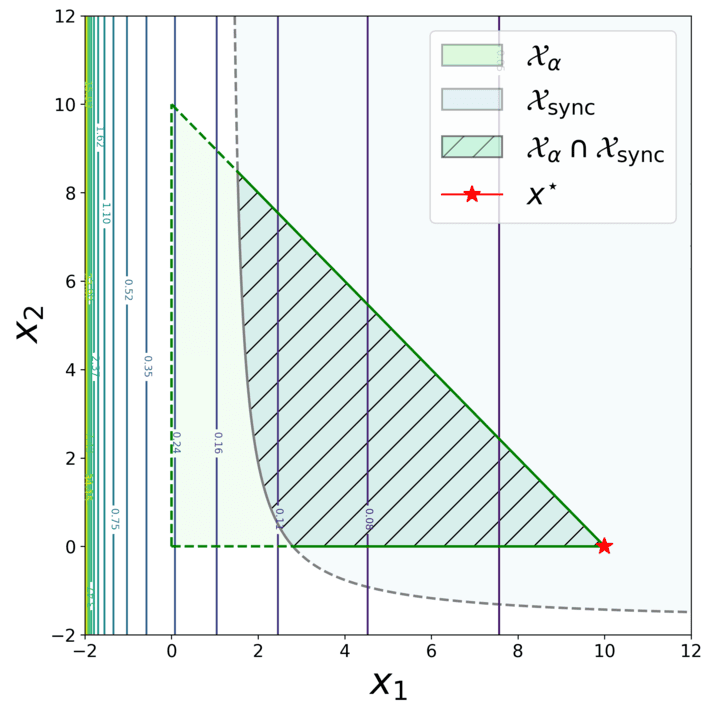}
        (a)
    \end{minipage}
    \hfill
    \begin{minipage}[t]{0.32\linewidth}
        \centering
        \includegraphics[width=\linewidth]{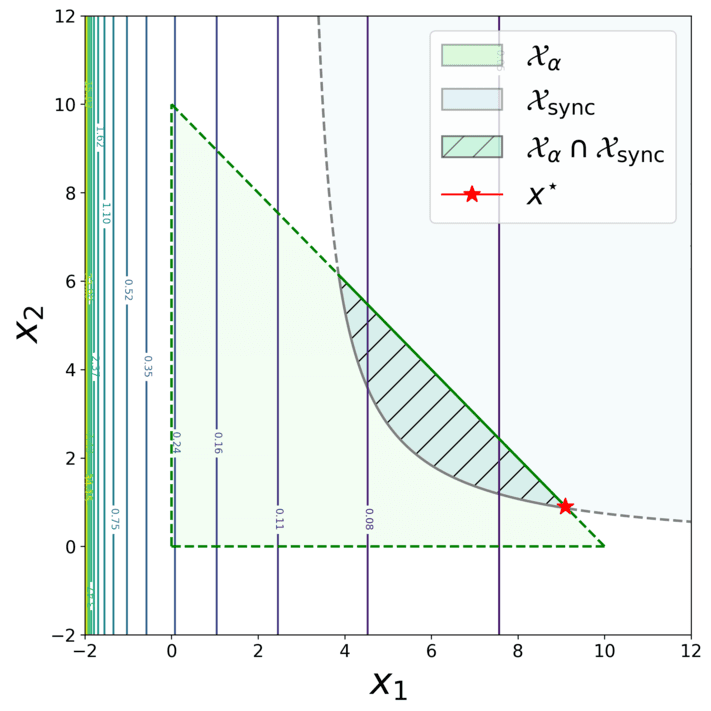}
        (b)
    \end{minipage}
    \hfill
    \begin{minipage}[t]{0.32\linewidth}
        \centering
        \includegraphics[width=\linewidth]{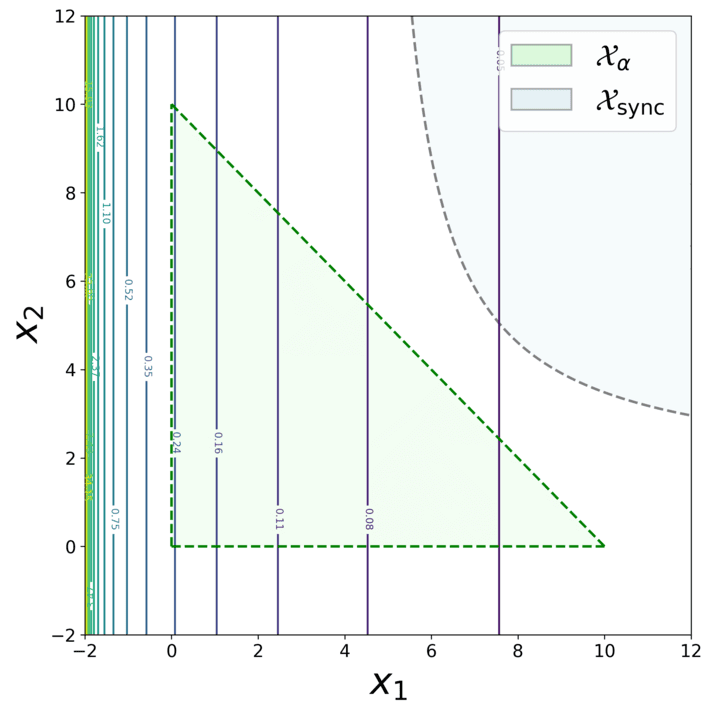}
        (c)
    \end{minipage}
    
    \caption{Level sets of the objective function $f(x) = \text{trace}(C^{\intercal}L_{\text{reg}, \beta}^{-1})$, the scaled-simplex $\mathcal{X}_{\alpha} = \left\{x\in \mathbb{R}_{\geq 0}^{2}\, \middle| \, 10 \geq  x^{\intercal}\mathbf{1}\right\}$, and $\mathcal{X}_{\text{sync}} = \left\{x\, \middle|\, \lambda_2(L(x)) \geq \sqrt{3 }\psi/\sin(\pi/4)\right\}$. (a) Feasible case (trivial solution): $\psi = 2.0$ p.u. (b) Feasible case (non-trivial solution): $\psi = 3.0$ p.u. (c) Infeasible case: $\psi = 4.0$ p.u.}
    \label{fig:simple_path_objective_function_and_feasible_sets}
\end{figure}

The matrix $C = \text{blkdiag}(\Pi_2, 0)$ and for a regularization parameter $\beta = 1$,
$L(x)_{\text{reg}, \beta} = L(x) + \frac{1}{3}\mathbf{1}\mathbf{1}^{\intercal}$. Fig.~\ref{fig:simple_path_objective_function_and_feasible_sets} above shows the level sets of the objective function and the feasible sets for $\alpha = 10$ p.u. and $\gamma = \pi/4$ rad and different values of $\psi$. The level sets of the objective function decay at a rate of $\mathcal{O}(1/x)$ along the $x_1$ direction and remain constant along the $x_2$ direction. Fig.~\ref{fig:simple_path_objective_function_and_feasible_sets}(b) shows that to guarantee $\gamma$-cohesiveness for $\psi = 3.0$ p.u., the constraint (\ref{eqn:LMI_1}) is active at the minimizer $x^{\star}$ and $x_2^{\star} > 0$ even though the branch $\{2,3\}$ does not connect generator nodes. For a fixed value of $\gamma$ and $\alpha \geq 0$, the parameter $\psi$ can be maximized using a one-dimensional binary search. This approach identifies the largest set of net power injections, $\mathcal{P}_{\text{max}}$, for which the optimization problem $P_2$ remains feasible.
A similar argument can be made for a fixed $\psi \geq 0$, while varying the arc length $\gamma \in (0, \pi/2)$.

\begin{corollary}
    A sufficient condition for the existence and uniqueness of a $\gamma$-cohesive synchronized solution $(\delta^*, \mathbf{0})$ for a power injection vector $p$ on the graph $\mathcal{G}(x)$ is that the pair-wise effective conductances satisfies
    \begin{align*}
        \frac{1}{r_{ij}^{\text{eff}}(\mathcal{G}(x))} \geq \frac{1}{2 \sin(\gamma)}\|B\|_2\|p\|_2\quad \forall \, i \neq j.
    \end{align*}
    \label{r_ij_synchronization_condition}
\end{corollary}
\begin{proof}
    Using the definition of effective resistance, 
    \begin{align*}
        r_{ij}^{\text{eff}}(\mathcal{G}(x)) &= (e_i - e_j)^{\intercal}L(x)^{\dagger}(e_i - e_j)\\
        &= (e_i - e_j)^{\intercal}V \Lambda^{\dagger} V^{\intercal}(e_i - e_j),\\
        &= y^{\intercal} \Lambda^{\dagger}y,
    \end{align*}
    where $y = V^{\intercal}(e_i - e_j)$. Taking norms and using the consistency of the vector and operator $2$-norm,
    \begin{align}
        \|r_{ij}^{\text{eff}}(\mathcal{G}(x))\|_2 = \|y^{\intercal} \Lambda^{\dagger}y\|_2 \leq \|y^{\intercal}\|_2\, \|\Lambda^{\dagger}\|_2 \, \|y\|_2 = \frac{2}{\lambda_2},
        \label{eqn:r_ij_lambda_2_inequality}
    \end{align}
    using the unitary invariance of the $2$-norm, since for all $i \neq j$, $y = \|V^{\intercal}(e_i - e_j)\|_2 = \|(e_i - e_j)\|_2 = \sqrt{2}$. Combining inequalities (\ref{eqn:r_ij_lambda_2_inequality}) and (\ref{eqn:synchronization_condition_2}), it follows that
    \begin{align*}
        \frac{1}{r_{ij}^{\text{eff}}(\mathcal{G}(x))} \geq \frac{\lambda_2(L(x))}{2} \geq \frac{1}{2 \sin(\gamma)}\|B\|_2\|p\|_2 \quad \forall i\neq j.
    \end{align*}
    is sufficient for $\gamma$-cohesiveness.
\end{proof}
Since $r_{ij}^{\text{eff}} = d(i, j)^2$ from Section~\ref{subsec:kron_reduction_and_effective_resistance}, the squared euclidean distance between any pair of nodes $i,j \in \mathcal{V}$ on the network $\mathcal{G}(x)$ is bounded above by $2/\lambda_2({\mathcal{G}(x))}$.

\section{Results}
\label{sec:results}
\subsection{Minimizing Transients}
\label{subsec:minimizing_transients}

In this section, we validate the framework developed in Section~\ref{sec:methodology} on the IEEE 30-bus test system and show that the problem formulated using the linearized model (\ref{eqn:linearized_dynamics}) results in significant improvements in the dynamics of the nonlinear differential-algebraic system (\ref{eqn:generator_model}) - (\ref{eqn:load_model}). The network data is available in MATPOWER \cite{zimmerman_matpower_2011}, the edge weights of the graphs $\mathcal{G}(x)$ in this section are the susceptances of the transmission lines, and the generator nodes are $\mathcal{V}_G = \{1, 2, 13, 22, 23, 27\}$. We model the problem using CVXPY \cite{diamond_cvxpy_2016} and solve with MOSEK \cite{aps_mosek_2025}. For the dynamic simulations in this section, we set the parameters of all generators to be $m = 1$ and $d = 1$; this choice is arbitrary, and all simulations are initialized at $\delta(t_0) =  \mathbf{0}$, $\dot{\delta}(t_0) = \mathbf{0}$.

We begin by comparing the synchronization dynamics for different strategies of allocating a budget $\alpha$. We consider the following methods:
\begin{enumerate}
    \item Proportional: $x_{e} = w_e \alpha/(\mathbf{1}^{\intercal}W \mathbf{1})$ for each $e \in \mathcal{E}$,
    \item Uniform: $x_e = \alpha/|\mathcal{E}|$ for each $e \in \mathcal{E}$,
    \item Random Uniform: $x_e = v_e\alpha$, $v \sim \text{Dirichlet}(1, 1, \ldots, 1)$,
    \item Optimal: $x_e = \arg \min P_1$.
\end{enumerate}

\begin{figure*}[htpb]
    \centering

    \begin{minipage}[t]{0.49\textwidth}
        \centering
        \includegraphics[width=\linewidth]{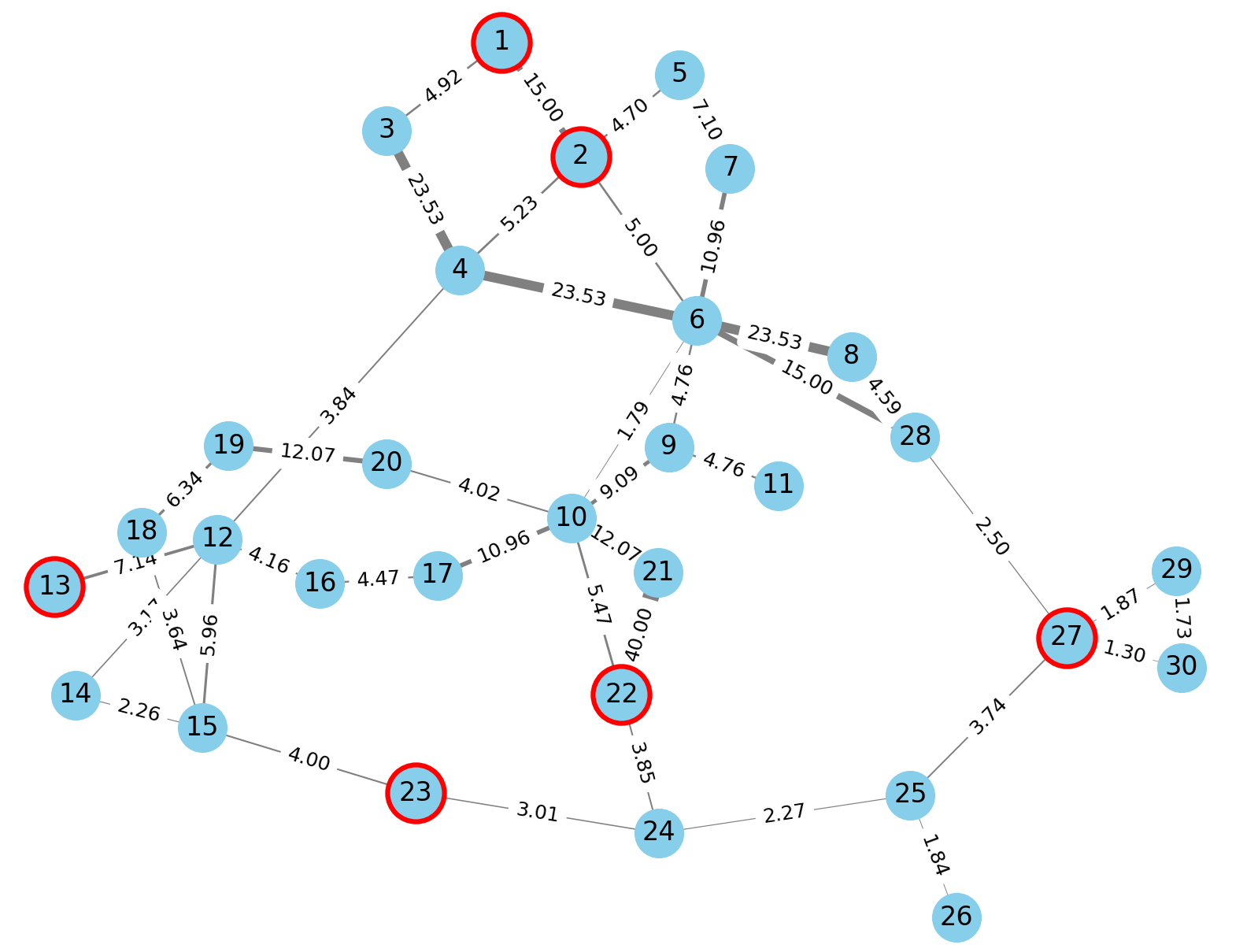}
        
        (a)
    \end{minipage}
    \hfill
    \begin{minipage}[t]{0.49\textwidth}
        \centering
        \includegraphics[width=\linewidth]{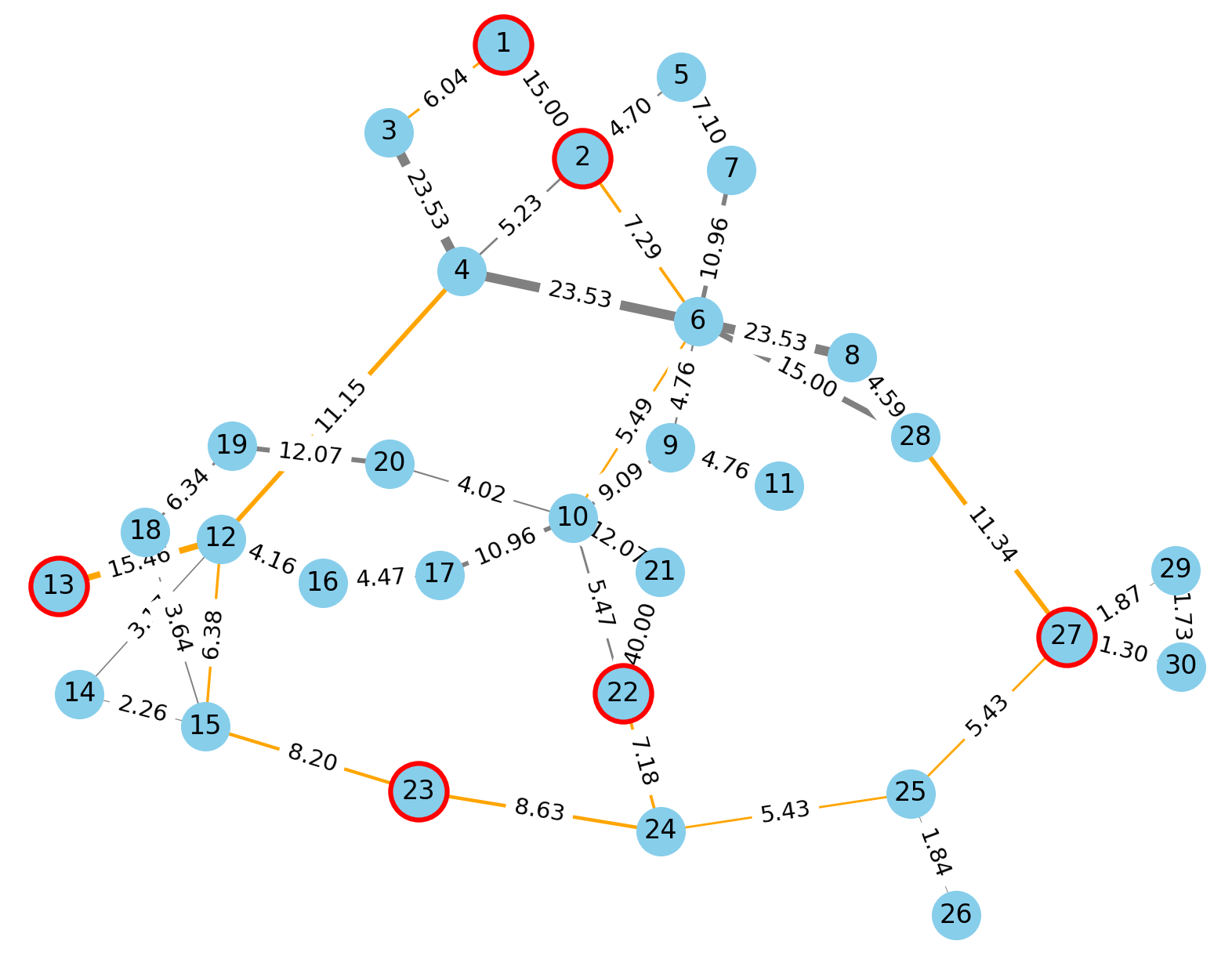}
        (b)
    \end{minipage}

    \caption{Comparing the edge weights of the IEEE 30-bus network before and after optimization. (a) Original network, $R_{\text{tot}}(\mathcal{G}_{\text{red}}) = 5.46$. (b) Network after optimization for $\alpha = 50$ p.u., $R_{\text{tot}}(\mathcal{G}_{\text{red}}(x^{\star})) = 2.91$.}
    \label{fig:comparing_ieee_30bus_pre_and_post_optimization}
\end{figure*}

\begin{figure}[!t]
    \centering

    \begin{minipage}[t]{0.32\linewidth}
        \centering
        \includegraphics[width=\linewidth]{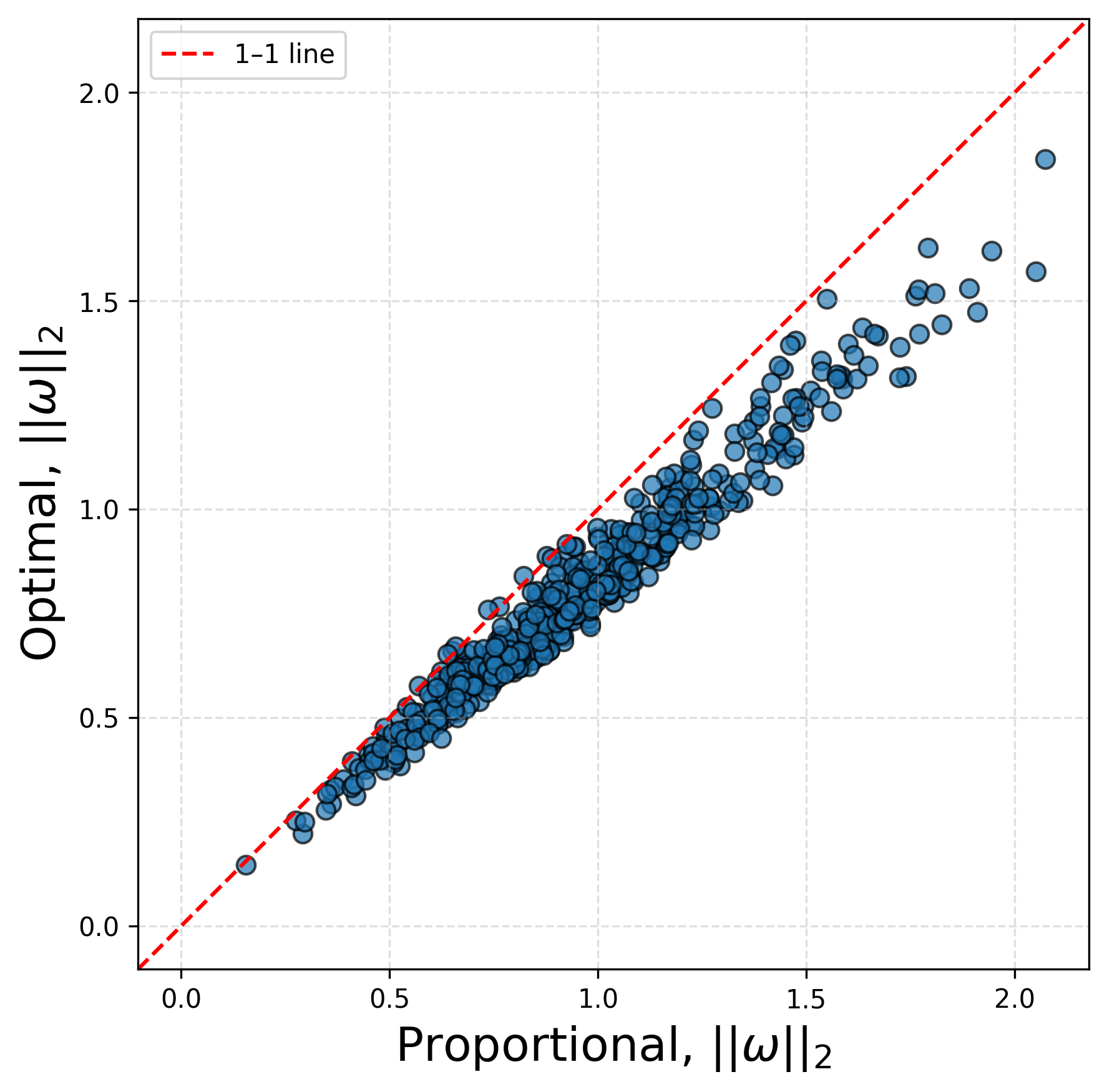}
        
        (a)
        \label{fig:optimal_vs_proportional}
    \end{minipage}
    \hfill
    \begin{minipage}[t]{0.32\linewidth}
        \centering
        \includegraphics[width=\linewidth]{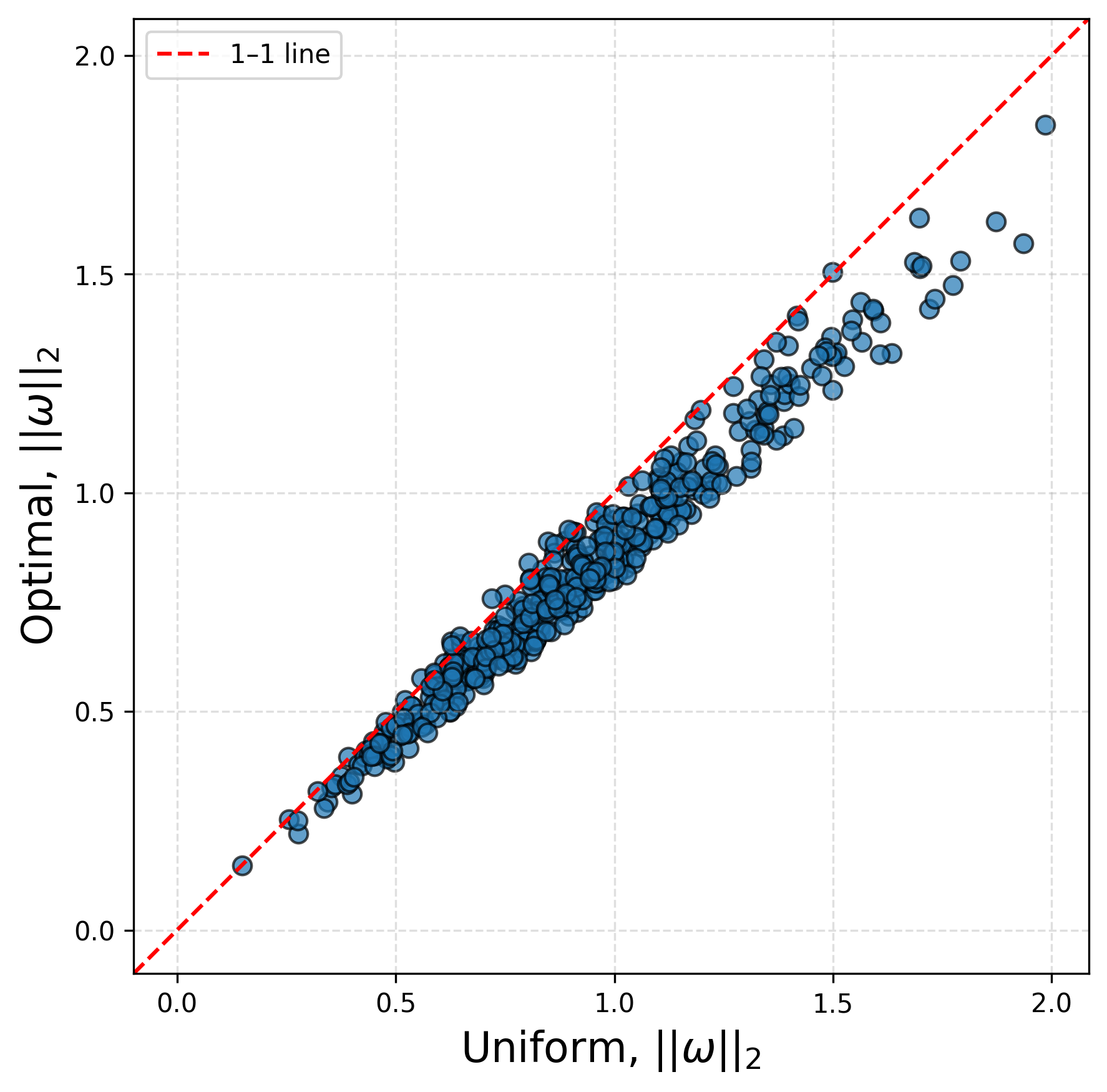}
        
        (b)
        \label{fig:optimal_vs_uniform}
    \end{minipage}
    \hfill
    \begin{minipage}[t]{0.32\linewidth}
        \centering
        \includegraphics[width=\linewidth]{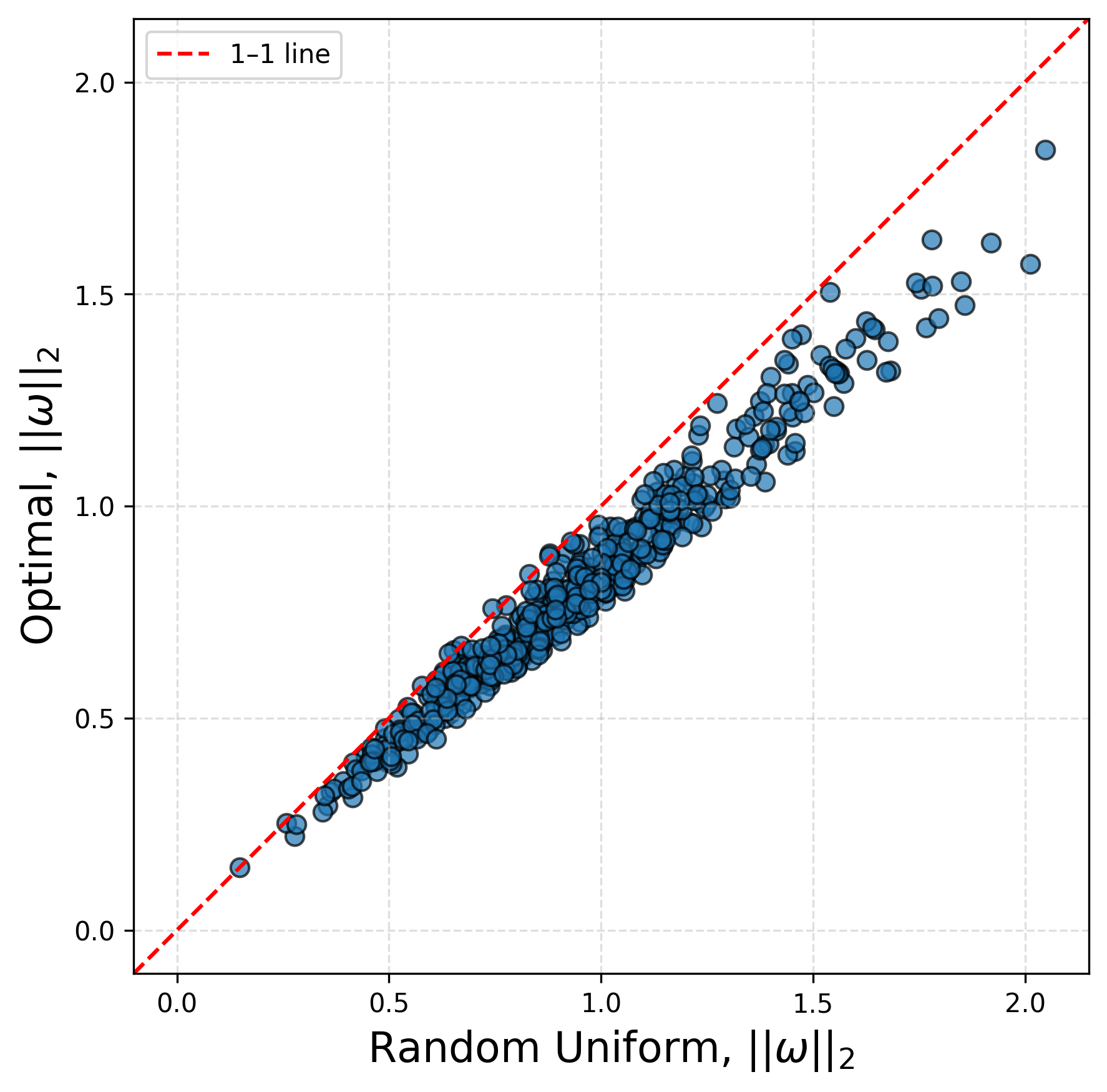}
        
        (c)
        \label{fig:optimal_vs_random_uniform}
    \end{minipage}

    \caption{Comparing $\|\omega\|_2$ at generator nodes for $\alpha = 50$ allocated using strategies (1) - (3) vs. the optimal method (4) for $500$ samples of exogenous inputs $u_0 \sim \mathcal{N}(0, I)$. The Optimal strategy outperforms (a) Proportional by an average of $15.79\%$. (b) Uniform by an average of $11.39\%$. (c) Random Uniform by an average of $14.64\%$.}
    \label{fig:comparing_optimal_to_all_strategies_for_alpha_50}
\end{figure}
To reduce the size of the graph $\mathcal{G}_{\text{red}}$ from $5.46$ to approximately $4.0$ for this network requires $\alpha \approx 15$ p.u using the optimal strategy, compared to $\alpha \approx 60$ p.u., $\alpha \approx 90$ p.u., and $\alpha > 100$ p.u. for the Uniform, Random Uniform, and Proportional strategies, respectively. We note that the results obtained using our method modifies only a subset of edges, in contrast to the other methods that modify all the branches of the network.

Considering a specific solution $x^{\star} = \arg\min P_1$ for $\alpha = 50$ p.u., the corresponding solution on the graph is shown in Fig.~\ref{fig:comparing_ieee_30bus_pre_and_post_optimization}. Comparing the original network Fig.~\ref{fig:comparing_ieee_30bus_pre_and_post_optimization}(a) and the optimized network Fig.~\ref{fig:comparing_ieee_30bus_pre_and_post_optimization}(b), we observe that most of $\alpha$ is allocated on the weakest links on paths connecting generator nodes, with a significant portion of $\alpha$ allocated to the leaf edge $\{12, 13\}$ connecting the lone generator node $13$ to the rest of the network. Also, as expected, branches connecting load nodes that do not lie on paths between generators (e.g., $\{5,7\}$, $\{25,26\}$, and $\{27,29\}$) are not modified. We simulate the DAE for the oscillator network defined on the graphs $\mathcal{G}(x)$ with edge weights modified according to methods (1) - (4) for $\alpha = 50$ p.u. Fig.~\ref{fig:comparing_optimal_to_all_strategies_for_alpha_50} compares the norms of the angular frequencies at the generator nodes for $500$ samples of exogenous inputs $u_0 \sim \mathcal{N}(0, I)$. Fig.~\ref{fig:comparing_optimal_to_all_strategies_for_alpha_50} shows that the optimal method results in a consistently better dynamic performance overall, especially for disturbances that result in larger transients. We remark that this improved dynamic performance is achieved by modifying the edge weights of $13$ of the $41$ edges on the graph $\mathcal{G}$ by allocating $\alpha = 50$ p.u. via the solution of $P_0$; all other aspects of the network remain identical. Simulations on the optimal graph $\mathcal{G}(x^{\star})$ also showed improved steady-state phase angle cohesion, even in the absence of an explicit phase cohesiveness constraint. This is consistent with the bound from Corollary~\ref{r_ij_synchronization_condition} using $r_{ij}^{\text{eff}}$ in Section~\ref{subsec:gamma_cohesiveness}.

\subsection{Minimizing Transients With Phase Angle Constraint}
\label{subsec:minimizing_transients_with_phase_angle_constraint}

Considering the problem of minimizing transients with an explicit synchronization constraint to ensure steady-state phase cohesion, Fig.~\ref{fig:ieee_30_optimized_with_gamma_constraint} shows the solution $x^{\star}$ to the problem $P_2$ for this network with parameters $\alpha = 50$ p.u., $\psi = 0.45$ p.u., and $\gamma = \pi/4$ rad. The optimal value for this problem is $\mathcal{G}_{\text{red}}(R_{\text{tot}}(x^{\star})) = 3.511$ compared to $\mathcal{G}_{\text{red}}(R_{\text{tot}}(x^{\star})) = 2.91$ obtained from solving $P_1$ (without the phase angle constraint).

Similar to Fig.~\ref{fig:comparing_ieee_30bus_pre_and_post_optimization}(b), the solution allocates $\alpha$ on the weakest links on the network. Fig.~\ref{fig:ieee_30_optimized_with_gamma_constraint} shows that to improve phase cohesion (via the constraint on $\lambda_2$), it is necessary to target the overall weakest links on the network $\mathcal{G}$ even if these links do not result in any improvements in the connectivity of $\mathcal{G}_{\text{red}}$; for instance, the highlighted connection to the leaf node $\{26\}$, and edges $\{27, 29\}$, $\{27, 30\}$. Consequently, the resulting solution trades off a degree of optimality relative to the solution from $P_1$, in order to improve steady-state phase angle cohesion. This observation is consistent with the simple path graph example in Fig.~\ref{fig:simple_path_objective_function_and_feasible_sets}(b).

\begin{figure}[h!]
    \centering
    \includegraphics[width=1.0\linewidth]{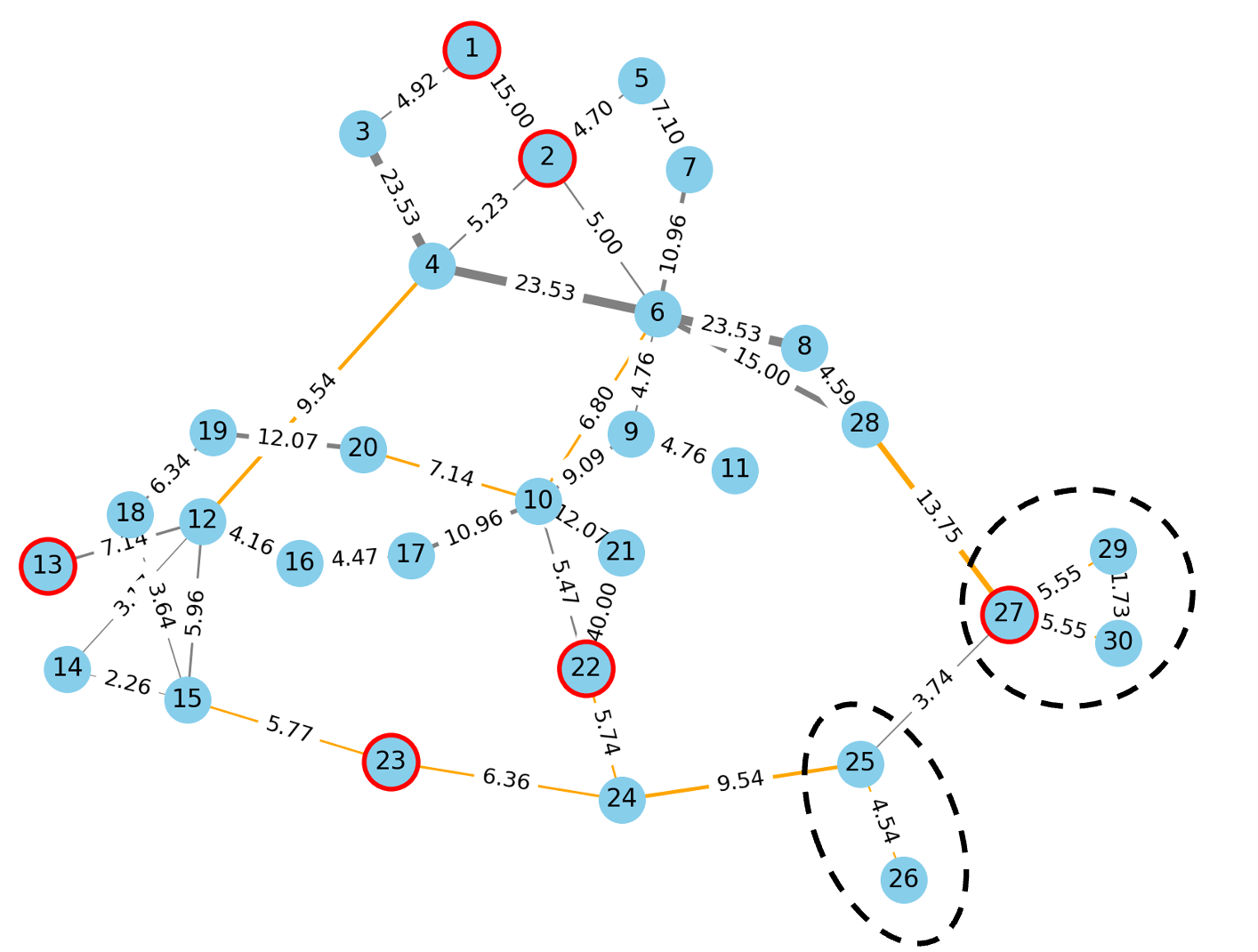}
    \caption{Optimal solution for $\alpha = 50$ with $\gamma$-cohesiveness constraint for $\gamma = \pi/4$ rad and $\psi = 0.45$ p.u. Optimal value, $R_{\text{tot}}(\mathcal{G}_{\text{red}}(x^{\star})) = 3.511$.}
    \label{fig:ieee_30_optimized_with_gamma_constraint}
\end{figure}

\subsection{Optimally Rewired Network}
\begin{figure}[!t]
    \centering
    \begin{minipage}[b]{0.49\textwidth}
        \centering
        \includegraphics[width=\linewidth]{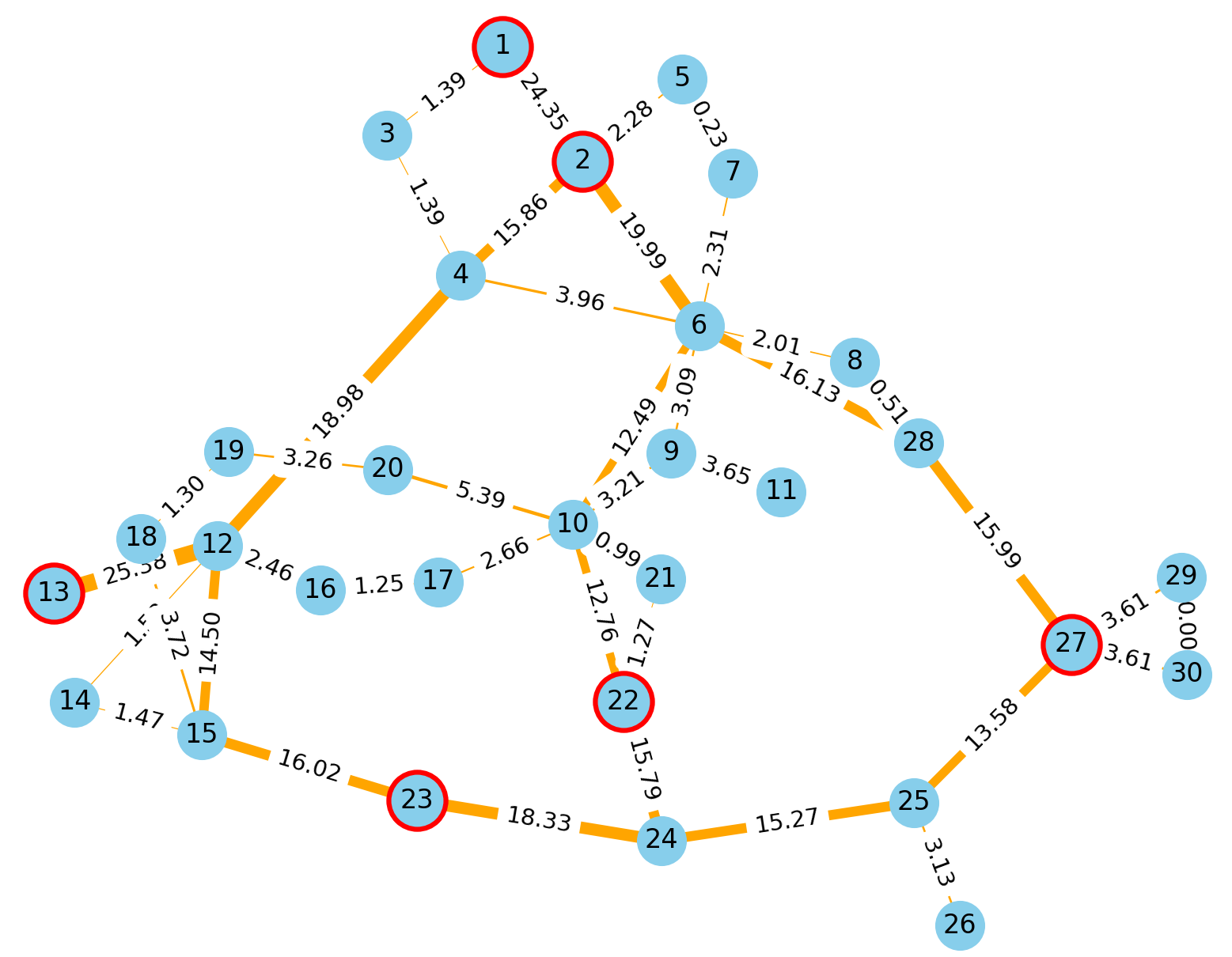}
        
        (a)
        
    \end{minipage}
    \hfill
    \begin{minipage}[b]{0.49\textwidth}
        \centering

        \begin{minipage}[t]{0.7\linewidth}
            \centering
            \includegraphics[width=\linewidth]{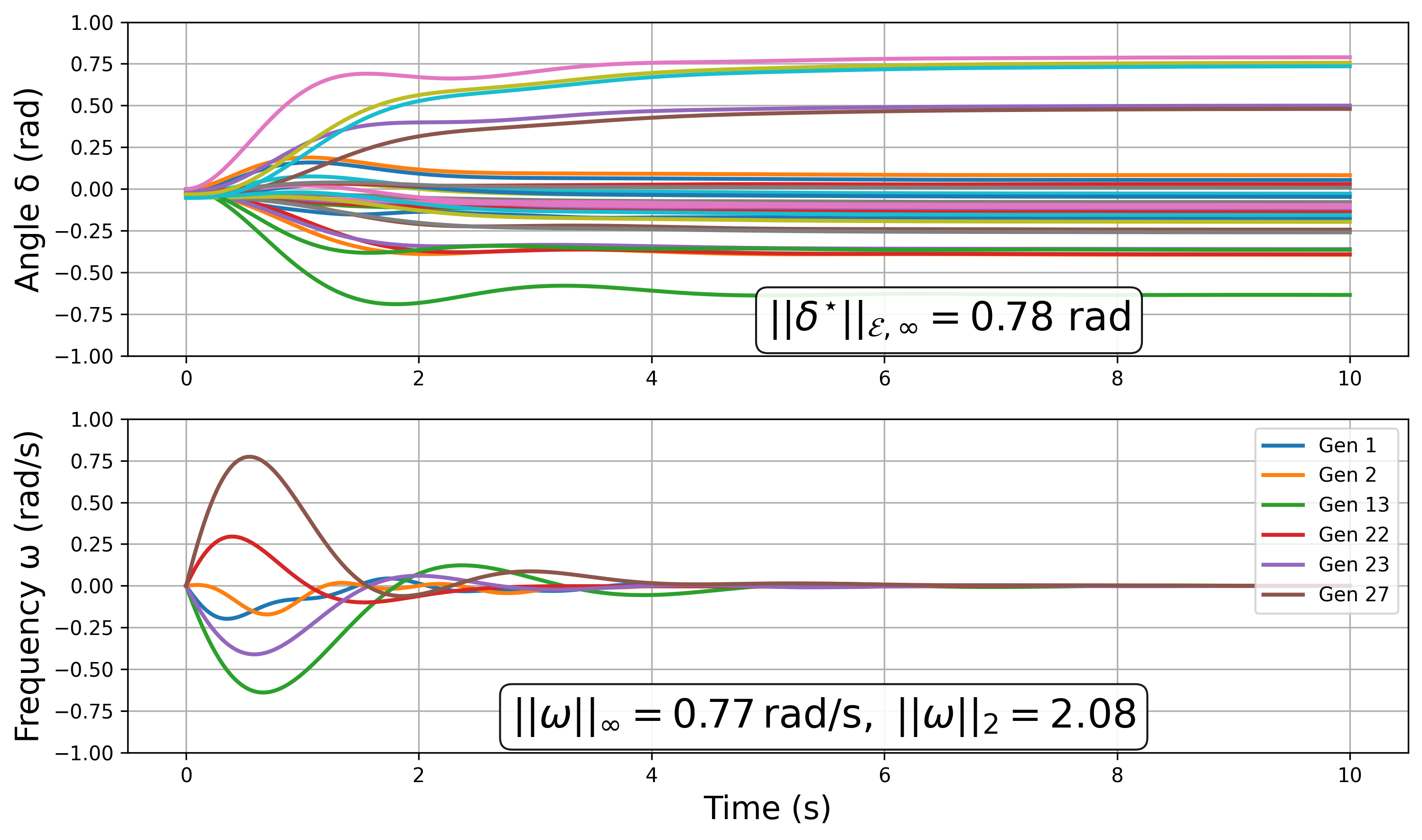}
            
            (b)
            
        \end{minipage}


        \begin{minipage}[t]{0.7\linewidth}
            \centering
            \includegraphics[width=\linewidth]{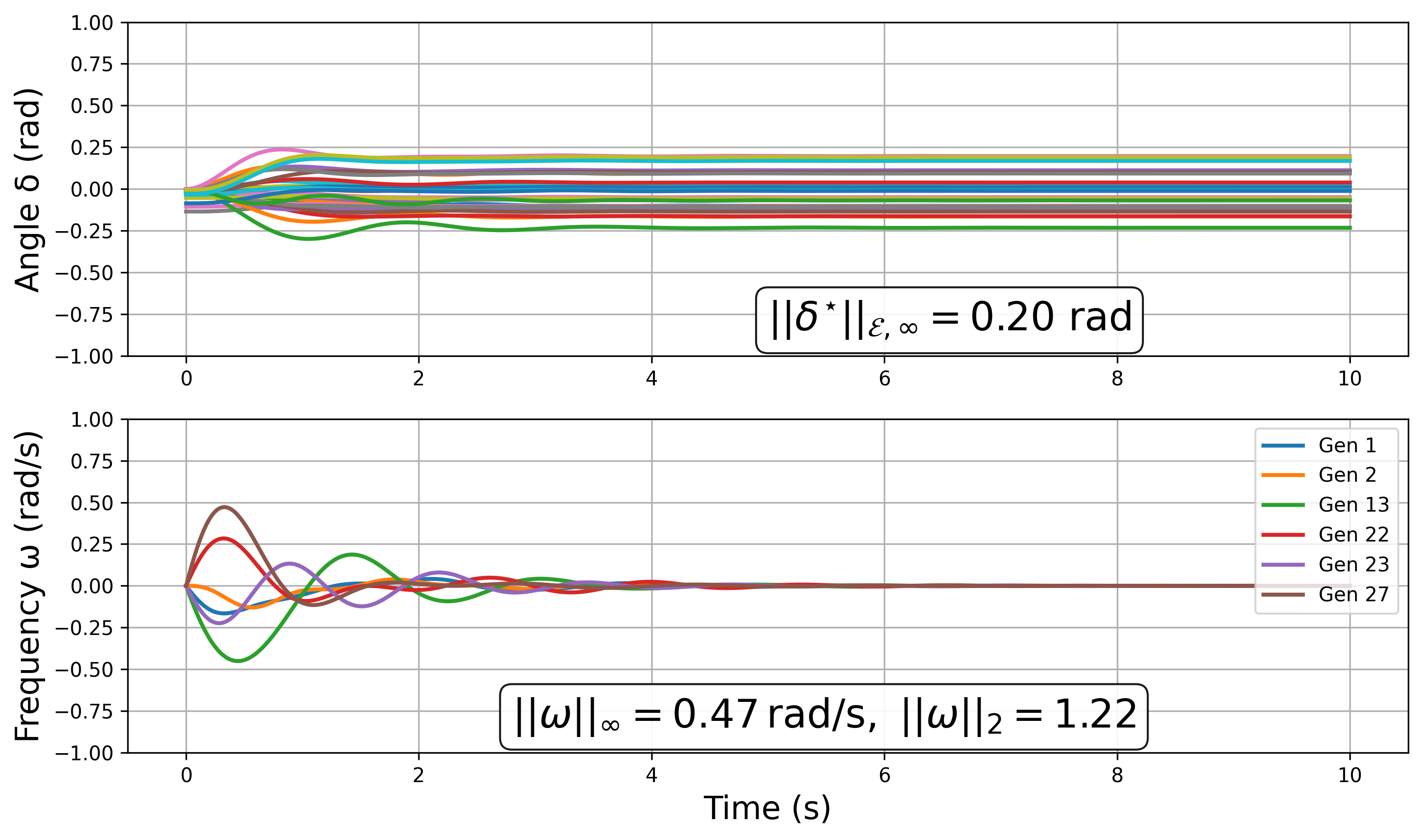}
            
            (c)
            
        \end{minipage}
    \label{fig:optimal_rewiring}

    \end{minipage}

    \caption{Optimally rewired network for parameters $\gamma = \pi/4$ rad and $\psi=0.45$ p.u., $R_{\text{tot}}(\mathcal{G}_{\text{red}}(x^{\star})) = 1.8954$ (a) Optimal graph.
    (b) Phase and frequency trajectories on the original network (Fig. 8(a)) for some realization of $u_0$.
    (c) Phase and frequency trajectories on optimally rewired network for the same $u_0$.}
    \label{fig:optimal_rewiring}
\end{figure}

Relaxing the non-negativity constraint on the design variable $x$ (that is, relaxing $X \succeq 0$) in the problem $P_2$ and setting $\alpha = 0$, we obtain the graph shown in Fig.~\ref{fig:optimal_rewiring}(a). The solution $x^{\star}$ to this problem reallocates the total edge weights $\mathbf{1}^{\intercal}W\mathbf{1}$ on the original graph $\mathcal{G}$ along the edges $\mathcal{E}$ to minimize the total effective resistance of the Kron-reduced graph subject to a phase angle constraint with parameters $\gamma = \pi/4$ rad $\psi = 0.45$ p.u. We note that, for power networks, there are considerations beyond synchronization that often necessitate retaining specific edges with prescribed edge weights in the network. In such cases, these edges should be removed from $\mathcal{E}_c$. However, if improving dynamic performance were the sole objective, Fig.~\ref{fig:optimal_rewiring} provides useful insights and could guide the design of oscillator networks. We remark that for this network reallocating edge weights also allows us to guarantee phase angle cohesiveness for much larger sets of net power injections $\mathcal{P}_{\psi}$.

Comparing phase angle trajectories in Fig.~\ref{fig:optimal_rewiring}(b) and Fig.~\ref{fig:optimal_rewiring}(c) for some realization of $u_0$ highlights the improved steady-state phase cohesion for the optimized network. For the same set of disturbances from Section~\ref{subsec:minimizing_transients}, Fig.~\ref{fig:comparing_original_network_to_optimally_rewired} shows that the rewired network exhibits significantly improved transient response. Overall, there's an average reduction of $36.19\%$ in $\|\omega\|_2$ with a reduction in the maximum incremental steady-state phase angles $\|\delta^{*}\|_{\mathcal{E}, \infty}$ from $0.78$ rad for the original network in Fig.~\ref{fig:comparing_ieee_30bus_pre_and_post_optimization}(a) to $0.20$ rad for Fig.~\ref{fig:optimal_rewiring}(a) across all realizations of $u_0$.

\begin{figure}[htbp]
    \centering
    \includegraphics[width=0.4\linewidth]{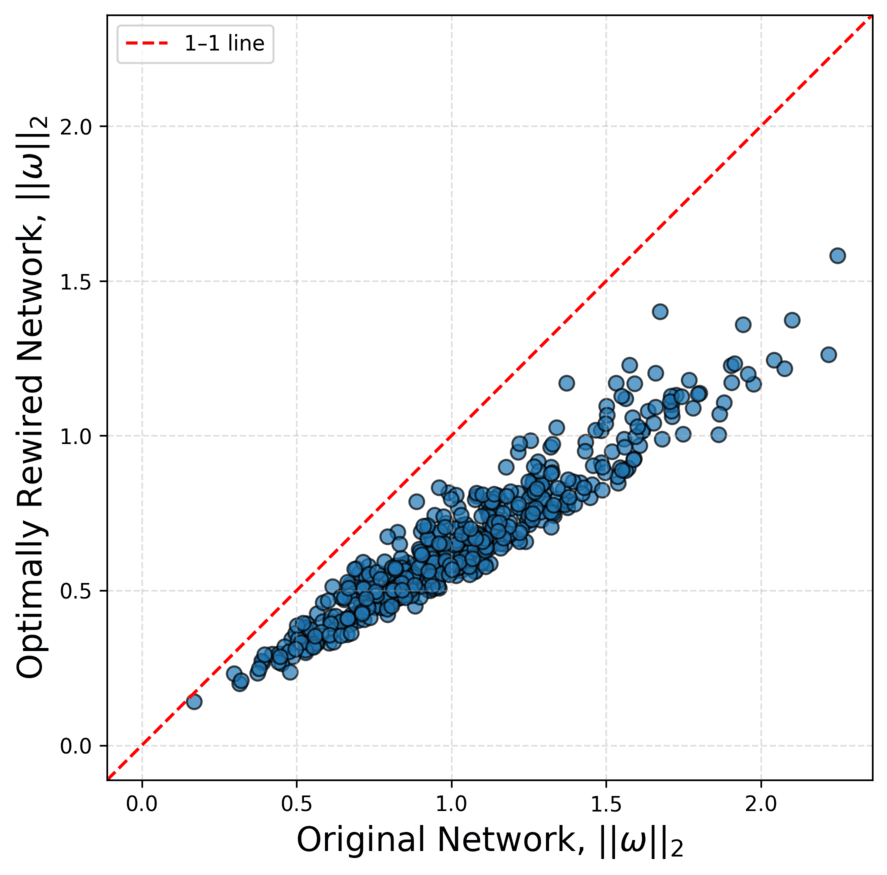}
    \caption{Comparing $\|\omega\|_2$ on the original network to the network with optimally reallocated total edge weights for the same samples of $u_0 \sim \mathcal{N}(0, I)$. The average reduction in frequency derivations is $36.19\%$.}
    \label{fig:comparing_original_network_to_optimally_rewired}
\end{figure}

\begin{figure}[t!]
    \centering
    \includegraphics[width=0.8\linewidth]{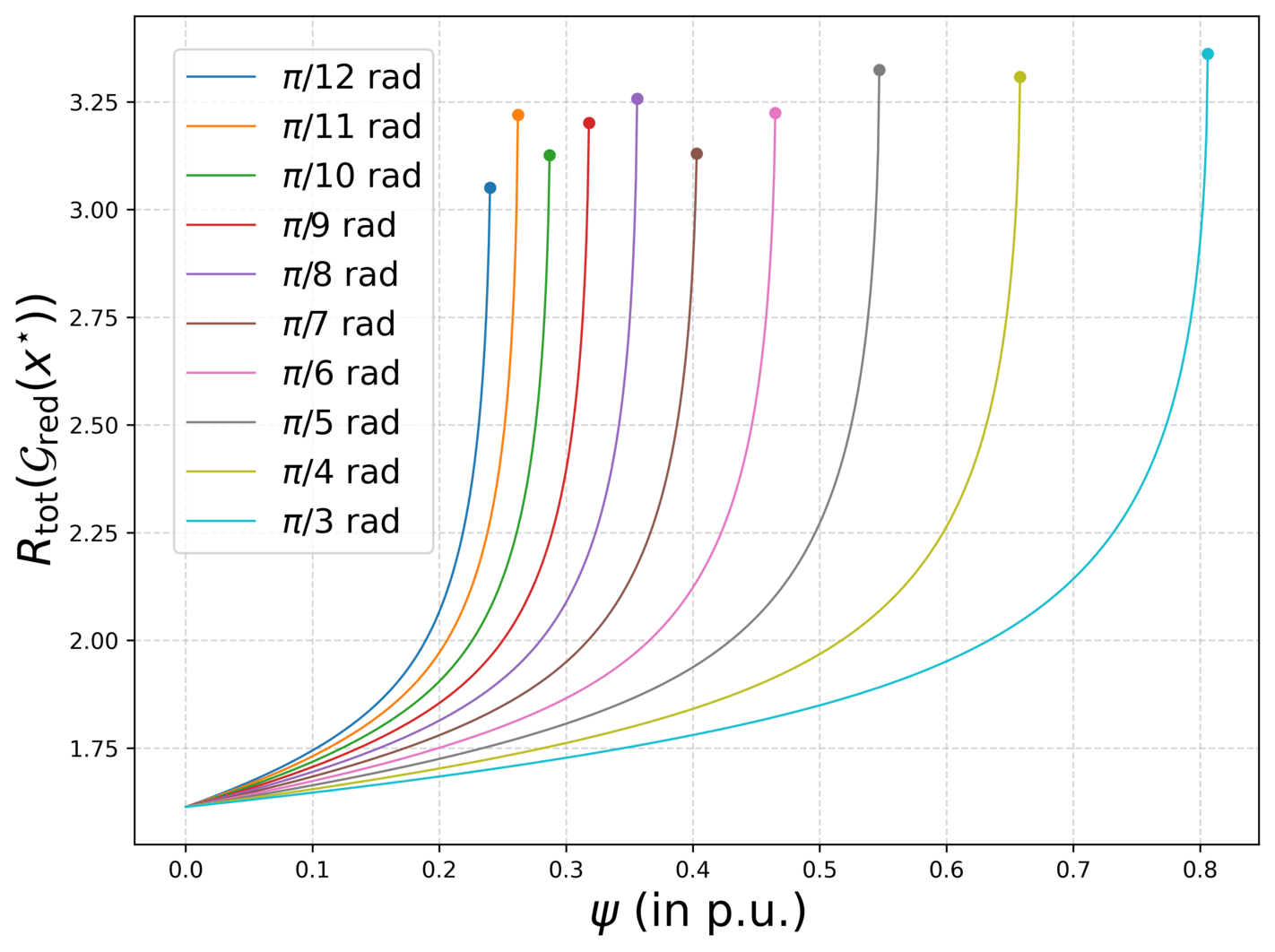}
    \caption{Optimal values vs. $\psi$ vs. $\gamma$ for the optimally rewired network (i.e. $\alpha =  0$ and relaxing $X \succeq 0$).}
    \label{fig:optimalval_vs_psi_vs_gamma}
\end{figure}

The relationship between the optimal value of the rewired network, $\psi$, and $\gamma$ for the optimal rewiring problem is illustrated in Fig.~\ref{fig:optimalval_vs_psi_vs_gamma}. For each fixed arc length $\gamma$, we increase the parameter $\psi$ from $0$ to the largest feasible value in step sizes of $10^{-3}$. Considering the choice of $\gamma  =\pi/4$ rad, we see the LMI is feasible up to $\psi = 0.658$ p.u., with an optimal value of $3.31$. From Fig.~\ref{fig:optimalval_vs_psi_vs_gamma}, the feasibility of the LMI (\ref{eqn:LMI_1}) for a given value of $\psi$ and $\gamma$ can be determined by tracing a vertical line from the $\psi$-axis; if the line does not intersect with the curve corresponding to $\gamma$, then the LMI (\ref{eqn:LMI_1}) is infeasible for that combination of $\psi$ and $\gamma$ as well as for all smaller values of $\gamma$.

\section{Conclusion}
In this work, we have shown that for a lossless power network of identical generators, the expected value of the $\mathcal{L}_2$ norm of the angular frequency dynamics near a frequency-synchronized solution can be decomposed into a transient component that is a function of the effective resistance of the underlying graph and a system-wide component that evolves in the nullspace of the reduced Laplacian (in $s$-domain), which is completely independent of the network parameters. This decomposition reveals the extent to which we can improve synchronization by modifying the network and highlights the primary factors influencing network transients. Most importantly, since the reduced Laplacian is always singular, there is a limit to the improvement in transient response that can be achieved by modifying the network.

We showed that under typical assumptions on the power network, the problem of minimizing the transient component of the frequency response at generator nodes by modifying the susceptances
of the network is a convex problem and we provide semidefinite programs for solving it. We further provide a sufficient condition, in terms of the algebraic connectivity, for the existence and uniqueness of a synchronized solution satisfying a specified steady-state phase angle requirement. A convex constraint, in the form of a linear matrix inequality, is also derived, whose feasibility guarantees that these requirements are met for a set of net power injections. Geometric interpretations of the solutions to the semidefinite programs are provided. We validate our proposed method on the IEEE 30-bus test system. Results for randomly generated exogenous inputs show that significant improvements in transient response can be achieved by optimally modifying the susceptances of the transmission network. Our proposed framework enables a systematic design of power networks as well as targeted interventions to improve synchronization.

\appendices



\section{Closed-form Expression For $\mathbb{E}_{u_0}[\|\widetilde{\omega}\|_2^2]$}
\label{Appendix_1}
For a step disturbance $u(t)$ defined as
\begin{align*}
    u(t) = \begin{cases}
        0, & t < 0,\\
        u_{0}\in \mathbb{R}^k, & t\geq 0,
    \end{cases}
\end{align*}
and $u_0(s) = u_0/s$. The corresponding angular frequency in Laplace domain with $u_0(s)$ as input,
\begin{align*}
    \omega(s) &= V H(s)V^{\intercal}\, \frac{u_0}{s}, \quad \text{where } \tilde{H}(s) = \frac{H(s)}{s}\\
    &= \begin{bmatrix} \frac{1}{\sqrt{k}}\mathbf{1} & V_{\perp}\end{bmatrix} \begin{bmatrix}
        \tilde{h}_1(s) & \\& \tilde{H}(s)_{\perp}
    \end{bmatrix}\begin{bmatrix}
        \frac{1}{\sqrt{k}} \mathbf{1}^{\intercal}\\ \\ V_{\perp}^{\intercal}
    \end{bmatrix}u_0\\
    &= \frac{1}{k} \tilde{h}_1(s)\mathbf{1}\mathbf{1}^{\intercal}u_0 + V_{\perp} \tilde{H}(s)_{\perp} V_{\perp}^{\intercal}u_0\\
    & \text{where } \frac{1}{k}\mathbf{1}^{\intercal}u_{0} = \frac{1}{k}\sum_{i}u_{0i} = \overline{u}_0 \in \mathbb{R}\\
    \omega(s) &= \underbrace{\tilde{h}_1(s) \overline{u}_0}_{\overline{\omega}(s)}\mathbf{1} + \underbrace{V_{\perp} \tilde{H}(s)_{\perp} V_{\perp}^{\intercal}u_0}_{\widetilde{\omega}(s)},\\
    &\text{where } \tilde{h}_1(s) = \frac{1}{s^2 m + sd}, \quad \text{since } \lambda_1 = 0
\end{align*}
$u_0 \in \mathbb{R}^{k}$ is a vector and $k$ is the number of generator nodes, $V_{\perp} \in \mathbb{R}^{k \times (k - 1)}$ is the matrix whose columns are the eigenvectors orthonormal to $\frac{1}{\sqrt{k}}\mathbf{1}$, and the matrix $\tilde{H}(s)_{\perp} = \text{diag}(\{\tilde{h}_i(s)\}_{i=2}^k) \in \mathbb{R}^{(k-1) \times (k-1)}$. In time domain the angular frequency can be similarly decomposed as,
\begin{align*}
    \omega(t) = \overline{\omega}(t) \mathbf{1} + \widetilde{\omega}(t).
\end{align*}
The $\mathcal{L}_2$ norm of the transient term
\begin{align*}
    \|\widetilde{\omega}\|_2^2 &= \int_{0}^{\infty}|\widetilde{\omega}(t)|^2 dt = \int_{0}^{\infty} \widetilde{\omega}(t)^{\intercal} \widetilde{\omega}(t) dt\\
    &= \int_{0}^{\infty}(V_{\perp} \tilde{H}(t)_{\perp} V_{\perp}^{\intercal}u_0)^{\intercal}(V_{\perp} \tilde{H}(t)_{\perp} V_{\perp}^{\intercal}u_0)dt\\
    &= \int_{0}^{\infty} u_0^{\intercal}V_{\perp} \tilde{H}(t)_{\perp}^2 V_{\perp}^{\intercal}u_0 dt\\
    & \text{define }z_0 = V_{\perp}^{\intercal}u_0 \in \mathbb{R}^{k-1}\\
    &= \int_{0}^{\infty} z_0^{\intercal} \tilde{H}(t)_{\perp}^2 z_0dt = \int_{0}^{\infty} \sum_{i = 1}^{k-1} z_{0i}^2 \tilde{h}_{i+1}(t)^2 dt\\
    &= \sum_{i = 1}^{k-1} z_{0i}^2 \int_{0}^{\infty} \tilde{h}_{i+1}(t)^2 dt = \sum_{i = 1}^{k-1} z_{0i}^2 \|\tilde{h}_{i + 1}\|_2^2
\end{align*}
where $\|\tilde{h}_i\|_2^2 = \frac{1}{2d \lambda_i}$ for $i = 2, \ldots, k$, by (28) of \cite{paganini_global_2020}.
\begin{align*}
    \|\widetilde{\omega}\|_2^2 &= \sum_{i = 1}^{k - 1} z_{0i}^2 \frac{1}{2d \lambda_{i + 1}}
    = \frac{1}{2d}\sum_{i = 1}^{k-1} \frac{(v_{i+1}^{\intercal}u_0)^2}{\lambda_{i+1}} = \frac{1}{2d}\sum_{i = 2}^{k} \frac{(v_{i}^{\intercal}u_0)^2}{\lambda_{i}}
\end{align*}
For $u_0 \sim \mathcal{N}(0, \sigma^2I)$, the expected value of the $\mathcal{L}_2$ norm of $\widetilde{\omega}$ has a closed-form expression, and we derive it as follows
\begin{align*}
    \mathbb{E}_{u_0}[\|\widetilde{\omega}\|_2^2] 
    &= \frac{1}{2d} \mathbb{E}_{u_0}\left[\frac{(v_2^{\intercal}u_0)^2}{\lambda_2} + \cdots + \frac{(v_k^{\intercal}u_0)^2}{\lambda_k}\right]\\
    &= \frac{1}{2d} \left(\frac{1}{\lambda_2} \mathbb{E}_{u_0}\left[(v_2^{\intercal}u_0)^2\right] + \cdots + \frac{1}{\lambda_k} \mathbb{E}_{u_0}\left[(v_k^{\intercal}u_0)^2\right] \right)\\
    &= \frac{1}{2d} \sum_{i=2}^k \frac{1}{\lambda_i} \mathbb{E}_{u_0}\left[(v_i^{\intercal}u_0)^2\right]
\end{align*}
Since $u_0 \sim \mathcal{N}(0, \sigma^2I)$, and $V \in \mathbb{R}^{k \times k}$ is orthonormal, this implies that $V^{\intercal}u_{0} \sim \mathcal{N}(0, \sigma^2I)$ and $\mathbb{E}_{u_0}\left[(v_i^{\intercal}u_0)^2\right] = \left({\mathbb{E}_{u_0}[v_i^{\intercal}u_0]}\right)^2 + \text{var}(v_i^{\intercal}u_0) = \sigma^2$ for each $i$.
So,
\begin{align*}
    \mathbb{E}_{u_0}[\|\widetilde{\omega}\|_2^2] &= \frac{\sigma^2}{2d} \sum_{i=2}^k \frac{1}{\lambda_i} = \frac{\sigma^2}{2d}\text{trace}\,\left(L_{\text{red}}^{\dagger}\right) = \frac{\sigma^2}{2d} \frac{1}{k} R_{\text{tot}} (\mathcal{G}_{\text{red}}),
\end{align*}
where $R_{\text{tot}}(\mathcal{G}_{\text{red}})$ is the total effective resistance of the Kron-reduced graph $\mathcal{G}_{\text{red}}$, $k$ is the number of generator nodes on the network, and $d$ is the damping coefficients of the oscillators.

\section{Proof of Lemma 2}
\label{appendix:invariance_of_r_ij_proof}
Partitioning the node voltages and current injections as
\begin{align*}
    v = \begin{bmatrix} v_{\mathcal{V}_G}^{\intercal} & v_{\overline{\mathcal{V}}_G}^{\intercal} \end{bmatrix}^{\intercal},
    \qquad
    J = \begin{bmatrix} J_{\mathcal{V}_G}^{\intercal} & J_{\overline{\mathcal{V}}_G}^{\intercal} \end{bmatrix}^{\intercal},
\end{align*}
and let $e_i$ denote the $i$th standard basis vector in $\mathbb{R}^k$. For a unit current injection at node $i \in \mathcal V_G$ and withdrawal at node $j \in \mathcal V_G$, the current injection vector
\begin{align*}
    J = \begin{bmatrix} (e_i - e_j)^{\intercal} & 0^{\intercal} \end{bmatrix}^{\intercal}.
\end{align*}
By the definition of effective resistance,
\begin{align*}
    r_{ij}^{\mathrm{eff}}(\mathcal G)
= v_i - v_j
= (e_i - e_j)^\intercal v_{\mathcal V_G}.
\end{align*}

The network equations are
\begin{align*}
    \begin{bmatrix}
    J_{\mathcal V_G} \\
    J_{\overline{\mathcal V}_G}
    \end{bmatrix}
    =
    \begin{bmatrix}
    L_{\mathcal V_G\mathcal V_G}
    &
    L_{\mathcal V_G\overline{\mathcal V}_G}
    \\
    L_{\overline{\mathcal V}_G\mathcal V_G}
    &
    L_{\overline{\mathcal V}_G\overline{\mathcal V}_G}
    \end{bmatrix}
    \begin{bmatrix}
    v_{\mathcal V_G} \\
    v_{\overline{\mathcal V}_G}
    \end{bmatrix}.
\end{align*}
Since $J_{\overline{\mathcal V}_G}=\mathbf 0$, eliminating
$v_{\overline{\mathcal V}_G}$
\begin{align*}
    J_{\mathcal V_G}
    =
    \Big(
    L_{\mathcal V_G\mathcal V_G}
    -
    L_{\mathcal V_G\overline{\mathcal V}_G}
    L_{\overline{\mathcal V}_G\overline{\mathcal V}_G}^{-1}
    L_{\overline{\mathcal V}_G\mathcal V_G}
    \Big)
    v_{\mathcal V_G}
    =
    L_{\mathrm{red}}v_{\mathcal V_G},
\end{align*}
where $L_{\overline{\mathcal V}_G\overline{\mathcal V}_G}$ is nonsingular for a connected graph,
\begin{align*}
    v_{\mathcal V_G}
    =
    L_{\mathrm{red}}^\dagger(e_i - e_j).
\end{align*}
Therefore,
\begin{align*}
    r_{ij}^{\mathrm{eff}}(\mathcal G)
    =
    (e_i - e_j)^{\intercal}
    L_{\mathrm{red}}^\dagger
    (e_i - e_j)
    =
    r_{ij}^{\mathrm{eff}}(\mathcal G_{\mathrm{red}}) \quad \forall i,j \in \mathcal{V}_G.
\end{align*}

\section{Sufficient Condition for the Existence and Uniqueness of a $\gamma$-cohesive Synchronized Solution}
\label{appendix:sufficient_sync_condition}
We build on the synchronization results in \cite{dorfler_synchronization_2013}, where the authors establish that for a vector of net power injections $p \in \mathbf{1}^{\perp}$, a unique and stable $\gamma$-cohesive solution $(\delta^*, \mathbf{0})$ exists for the power network equation (\ref{eqn:generator_model}) - (\ref{eqn:load_model}) on the graph $\mathcal{G}(x)$ if
\begin{align*}
    \|L(x)^{\dagger}p\|_{\hat{\mathcal{E}}, \infty} \leq \sin(\gamma),
\end{align*}
where the incremental $\infty$-norm is defined for a vector $x$ as $\|x\|_{\hat{\mathcal{E}}, \infty} := \max_{\{i, j\} \in \hat{\mathcal{E}}} |x_i - x_j|$ and $L(x)$ is the Laplacian of $\mathcal{G}(x)$. For the edge set $\hat{\mathcal{E}} = \mathcal{E} \cup \mathcal{E}_c$ with incidence matrix $\hat{B}$, the norm can be written in terms of the incidence matrix as $\|\hat{B}^{\intercal} x\|_{\infty}$. For any $p \in \mathbf{1}^{\perp}$ and $\gamma \in (0, \pi/2)$, assuming the eigenvalues of $L(x)$ are ordered as $0 = \lambda_1 <\lambda_2 \leq \ldots \leq \lambda_n$, then
\begin{align*}
    \|L(x)^{\dagger}p\|_{\hat{\mathcal{E}}, \infty} &= \|\hat{B}^{\intercal}L(x)^{\dagger}p\|_{\infty}\\
    &= \|\hat{B}^{\intercal}V \text{diag}(0, 1/\lambda_2, \ldots, 1/\lambda_n) V^{\intercal} p\|_{\infty}\\
    &= \frac{1}{\lambda_2}\|\hat{B}^{\intercal}V \text{diag}(0, 1, \lambda_2/\lambda_3 \ldots, \lambda_2/\lambda_n) V^{\intercal} p\|_{\infty}\\  
    &\leq \frac{1}{\lambda_2}\|\hat{B}^{\intercal}V \text{diag}(0, 1, \lambda_2/\lambda_3 \ldots, \lambda_2/\lambda_n) V^{\intercal} p\|_2\\
    &\text{using the consistency of the $2$-norm,}\\
    &\leq \frac{1}{\lambda_2}\|\hat{B}\|_2 \|V\text{diag}(0, 1, \lambda_2/\lambda_3 \ldots, \lambda_2/\lambda_n) V^{\intercal} p\|_2\\
    &\leq \frac{1}{\lambda_2} \|\hat{B}\|_2 \|V \text{diag}(0, 1, \ldots, 1) V^{\intercal} p\|_2\\
    &= \frac{1}{\lambda_2} \|\hat{B}\|_2 \|(I - \frac{1}{n}\mathbf{11}^{\intercal})p\|_2 = \frac{1}{\lambda_2} \|\hat{B}\|_2 \|p\|_{2}.
\end{align*}
It follows that a sufficient condition for $\gamma$-cohesiveness is,
\begin{align*}
    \|\hat{B}\|_2 \|p\|_2 \frac{1}{\sin(\gamma)} \leq \lambda_2,
\end{align*}
where the operator norm $\|\hat{B}\|_2 = \sqrt{\lambda_n(BB^{\intercal})} \leq \min \left\{ \sqrt{n}, 2d_{\text{max}}\right\}$.
\section{Proof of Proposition~\ref{C_is_projector}}
\label{appendix:c_is_projector}
$C$ is an orthogonal projector iff $C = C^{\intercal}$ and $C^2 = C$, that is, $C$ is symmetric and idempotent. $C$ is symmetric by definition. We show that it is idempotent as follows.
Recall that $E$ is orthonormal by definition, that is, $E^{\intercal}E = I$,
\begin{align*}
    C^2 &= \left(EE^{\intercal} - \frac{1}{k} \mathbf{1}_{\mathcal{V}_G}\mathbf{1}_{\mathcal{V}_G}^{\intercal}\right) \left(EE^{\intercal} - \frac{1}{k} \mathbf{1}_{\mathcal{V}_G}\mathbf{1}_{\mathcal{V}_G}^{\intercal}\right)\\
    &= EE^{\intercal} - \frac{2}{k}(EE^{\intercal} \mathbf{1}_{\mathcal{V}_G}) \mathbf{1}_{\mathcal{V}_G}^{\intercal} + \frac{k}{k^2}\mathbf{1}_{\mathcal{V}_G}\mathbf{1}_{\mathcal{V}_G}^{\intercal}\\
    &\text{Since } \mathbf{1}_{\mathcal{V}_G} \in \text{Span}(EE^{\intercal}) \text{ that is } EE^{\intercal} \mathbf{1}_{\mathcal{V}_G} = \mathbf{1}_{\mathcal{V}_G},\\
    &= EE^{\intercal} - \frac{1}{k}\mathbf{1}_{\mathcal{V}_G}\mathbf{1}_{\mathcal{V}_G}^{\intercal} = C.
\end{align*}
Hence $C$ is a projector, and we can verify that it's a projector onto the $(k-1)$-dimensional subspace defined as $\left\{x \;\left| \; \mathbf{1}_{\mathcal{V}_G}^{\intercal} x = 0,\right.\; x_j = 0 \;\forall j \notin \mathcal{V}_G\right\}$.



\bibliographystyle{IEEEtran}

\bibliography{zotero_references}
\vspace{-1em}

\begin{IEEEbiography}{Gerald Ogbonna}
is a PhD candidate in the Systems Engineering department at Cornell University. He received a B.Sc. in Electrical and Electronics Engineering from Federal University of Technology Owerri, Nigeria.

His research explores power network dynamics through the framework of coupled oscillator models, as part of the broader problem of the control of networked multi-agent systems.
\end{IEEEbiography}
\vspace{-1em}

\begin{IEEEbiography}{David Bindel} received BS degrees in mathematics and computer science from the University of Maryland in 1999, and a PhD in computer science from UC Berkeley in 2006. After three years as a Courant Instructor of mathematics at NYU, he joined the department of Computer Science at Cornell University, where he is currently a professor of Computer Science and the director of the Center for Applied Mathematics (CAM). He also serves as the director of the Simons Collaboration on Hidden Symmetries and Fusion Energy. His research focus is in applied numerical linear algebra and scientific computing, with applications to a variety of science and engineering problems. He is a fellow of the Society for Industrial and Applied Mathematics (SIAM) and the recipient of the Householder Prize in numerical linear algebra, a Sloan research fellowship, and best paper awards from the KDD and ASPLOS conferences and from the SIAM Activity Group on Linear Algebra.
\end{IEEEbiography}
\vspace{-1em}

\begin{IEEEbiography}{C. Lindsay Anderson}
 is Professor and Chair of Biological and Environmental Engineering at Cornell University, with research affiliations in the Center for Applied Mathematics, Systems Engineering, and Electrical and Computer Engineering. Previously, she served as the Kathy Dwyer Marble and Curt Marble Faculty Director at the Cornell Atkinson Center for Sustainability and as interim Director of the Cornell Energy Systems Institute.  Her research interests are the application of optimization under uncertainty to large-scale problems in sustainable energy systems. The National Science Foundation, US Department of Energy, US Department of Agriculture, PSERC, and the National Science and Engineering Research Council of Canada have supported her work. 
\end{IEEEbiography}

\end{document}